\documentclass[conference]{IEEEtran}
\IEEEoverridecommandlockouts

\usepackage{hyperref}
\usepackage{cite}
\usepackage{amsmath,amssymb,amsfonts}
\usepackage{algorithmic}
\usepackage{graphicx}
\usepackage{textcomp}
\usepackage{xcolor}
\usepackage{cleveref}
\usepackage{hhline}
\usepackage[linesnumbered,ruled,noend]{algorithm2e}
\usepackage{amsthm}
\usepackage{bm}
\usepackage{thm-restate}
\usepackage{pifont}
\usepackage{graphicx}
\usepackage{subcaption}
\usepackage{rotating}
\usepackage[T1]{fontenc}

\usepackage{newtxtext,newtxmath}

\usepackage{xspace}
\usepackage{comment}
\usepackage{tikz}
\usepackage{enumitem}
\usepackage{multirow}
\usepackage{multicol}
\usepackage{xurl} 
\usepackage{xcolor}
\definecolor{cssblue}{HTML}{0000FF} 
\def\BibTeX{{\rm B\kern-.05em{\sc i\kern-.025em b}\kern-.08em
    T\kern-.1667em\lower.7ex\hbox{E}\kern-.125emX}}

\allowdisplaybreaks[4]

\begin{document}

\title{Approximate Butterfly Counting in Sublinear Time\\
}

\author{
    \IEEEauthorblockN{Chi Luo$^{1*}$, Jiaxin Song$^{2*}$, Yuhao Zhang$^{3}$, Kai Wang$^1$, Zhixing He$^3$, Kuan Yang$^3$}
    \IEEEauthorblockA{
     	$^1$Data-Driven Management Decision Making Lab, Antai College of Economics and Management, Shanghai Jiao Tong University\\
            $^2$University of Illinois, Urbana-Champaign, 
            $^3$John Hopcroft Center, Shanghai Jiao Tong University,\\
 		parkourfr4ever@gmail.com,  jiaxins8@illinois.edu, \{zhang\_yuhao, w.kai, chandlr, kuan.yang\}@sjtu.edu.cn
 	}
    \thanks{$^{*}$ Equal contributions.}
    \thanks{ Code is available at: \url{https://github.com/PKFR4ever/Approximate-Butterfly-Counting-in-Sublinear-Time}}
}

\newtheorem{theorem}{Theorem}
\newtheorem{proposition}{Proposition}
\newtheorem{definition}{Definition}
\newtheorem{lemma}[theorem]{Lemma}
\newtheorem{corollary}[theorem]{Corollary}
\newtheorem{remark}[theorem]{Remark}
\newtheorem{assumption}[lemma]{Assumption}
\newtheorem{fact}[lemma]{Fact}
\newtheorem{claim}[lemma]{Claim}
\newtheorem{example}{\textbf{Example}}
\newcommand{\btf}{\rotatebox[origin=c]{90}{$\bowtie$}}
\newcommand{\todo} [1]{\textcolor{purple}{{\sf TODO}: #1}}
\renewcommand{\Pr}{\bm{\mathrm{Pr}}}
\newcommand{\E}{\bm{\mathrm{E}}}
\newcommand{\Var}{\bm{\mathrm{Var}}}
\newcommand{\etal}{\emph{et al.}}
\SetKwInput{KwInput}{Input}                
\SetKwInput{KwOutput}{Output} 
\SetKwFunction{Heavy}{{\sc Heavy}} 
\newcommand{\Espar}{{\sc ESpar}}

\newcommand{\wt}[1]{\mathrm{wt}{#1}}
\newcommand{\Comment}[1]{\hfill {\it \color{gray} $\triangleright$ #1}}
\newcommand{\abs}[1]{\left\vert#1\right\vert}

\newcommand{\numds}{15}
\newcommand{\jx}[1]{{\color{red} Jiaxin: #1}}
\newcommand{\hzx}[1]{{\color{red} hezhixing: #1}}
\newcommand{\lc}[1]{\color{blue} {#1}}

\definecolor{green(html/cssgreen)}{rgb}{0.0, 0.5, 0.0}

\definecolor{darkgray}{rgb}{0.66, 0.66, 0.66}
\definecolor{dimgray}{rgb}{0.41, 0.41, 0.41}

\newcommand{\codecmt}[1]{{\it \small \color{green(html/cssgreen)} $\triangleright$ #1}}

\newcommand{\naive}{${\tt ESpar}$\xspace}
\newcommand{\btfe}{${\tt TLS\text{-}EG}$\xspace} 
\newcommand{\heavy}{${\tt Heavy}$\xspace}
\newcommand{\btff}{${\tt TLS\text{-}HL\text{-}GP}$\xspace} 
\newcommand{\bs}{${\tt WPS}$\xspace} 
\newcommand{\our}{${\tt TLS}$\xspace} 

\SetKwInput{KwInput}{Input}                
\SetKwInput{KwOutput}{Output}


\maketitle

\begin{abstract}
Bipartite graphs serve as a natural model for representing relationships between two different types of entities. When analyzing bipartite graphs, \emph{butterfly counting} is a fundamental research problem that aims to count the number of butterflies (i.e., 2$\times$2 bicliques) in a given bipartite graph. While this problem has been extensively studied in the literature, existing algorithms usually necessitate access to a large portion of the entire graph, presenting challenges in real scenarios where graphs are extremely large and I/O costs are expensive. In this paper, we study the butterfly counting problem under the query model, where the following query operations are permitted: \emph{degree query}, \emph{neighbor query}, and \emph{vertex-pair query}.
We propose \our, a practical two-level sampling algorithm that can estimate the butterfly count accurately while accessing only a limited graph structure, achieving significantly lower query costs under the standard query model. \our also incorporates several key techniques to control the variance, including ``small-degree-first sampling'' and ``wedge sampling via small subsets''. 
To ensure theoretical guarantees, we further introduce two novel techniques: ``heavy-light partition'' and ``guess-and-prove'', integrated into \our. 
With these techniques, we prove that the algorithm can achieve a $(1+\epsilon)$ accuracy for any given approximation parameter $0 < \epsilon < 1$ on general bipartite graphs with a promised time and query complexity. In particular, the promised time is sublinear when the input graph is dense enough. Extensive experiments on {15} datasets demonstrate that \our delivers robust estimates with up to three orders of magnitude lower query costs and runtime compared to existing solutions. 
\end{abstract}

\section{Introduction}
\label{sec:intro}
Bipartite graphs serve as a natural model for representing relationships between two different types of entities, such as author-paper, user-item, and actor-movie. As the minimum complete subgraph in bipartite graphs, butterfly, which is also known as 2$\times$2 bicliques, plays an important role in bipartite graph analytics. { In the field of bipartite graph analytics, butterfly counting is a fundamental research problem that aims to count the number of butterflies on a given bipartite graph. For an instance of butterfly counting, in Figure \ref{fig:example0}, the vertices $\{u_0,u_1,v_0,v_1\}$ form a butterfly, and there are two butterflies in this graph. 
Butterfly counting has been proven useful in many applications, including network modeling \cite{Aksoy2016MeasuringAM, Lind2005CyclesAC, Opsahl2010TriadicCI, Robins2004SmallWA}, community structure mining \cite{Saryce2016PeelingBN, Zou2016BitrussDO}, and network security \cite{Li2021ApproximatelyCB}. Representative application scenarios include the following.
%
1) {\em Network Measurement}. The bipartite clustering coefficient is a cohesiveness measurement of bipartite graphs \cite{DBLP:journals/compnet/AksoyKP17, lind2005cycles, DBLP:journals/socnet/Opsahl13, DBLP:journals/cmot/RobinsA04, DBLP:conf/complexis/VukicJC23, DBLP:journals/vldb/PanZLYLGL25}. 
For instance, in user-item networks, the bipartite clustering coefficient can be used to identify sticky interest groups, which are useful for group-based recommendation \cite{su2009survey}. 
Given a bipartite graph $G$, its bipartite clustering coefficient is defined as $4 \times \btf / n_G$, where \btf\xspace  is the number of butterflies and $n_G$ denotes the number of caterpillars (i.e., three-paths) in $G$. Since $n_G$ can be computed in $O(m)$ time \cite{DBLP:journals/compnet/AksoyKP17}, accelerating butterfly counting will in turn accelerate the computation of the bipartite clustering coefficient. 

\noindent 2) {\em Fraud Detection}. Butterfly counting is also widely used in algorithms discovering large and dense subgraphs in massive graphs~\cite{DBLP:conf/vldb/GibsonKT05,DBLP:journals/vldb/WangLQZZ22,DBLP:journals/tkde/WangZZQZ23}. 
For example, many dense subgraphs in the World Wide Web’s host-connection graph correspond to \emph{link spam}~\cite{DBLP:conf/vldb/GibsonKT05} (i.e., websites that attempt to manipulate search engine rankings through aggressive interlinking to simulate popular content.), making dense subgraph extraction a useful primitive for spam detection in search engines. In e-commerce fraud detection, butterfly counting has further been applied to prune infeasible vertices~\cite{DBLP:journals/pvldb/LyuQLZQZ20}.
}









In recent years, this problem has been extensively studied in the literature \cite{Wang2014RectangleCI, sanei2018butterfly, Wang2018VertexPB}. However, existing algorithms often require access to a large portion of the entire graph, which poses challenges in real-world scenarios where graphs are massive and I/O costs are prohibitively high. For instance, in the cloud computing environment, graphs can be too large to store locally. In such scenarios, network bandwidth often becomes a critical performance bottleneck \cite{tan2021choosing}, and accessing input data incurs substantial overhead. To address this challenge, designing a sublinear-time algorithm that avoids reading the entire graph is essential.
The sublinear algorithm literature has employed the query model to analyze graph access costs for various problems, including triangle counting \cite{Eden2015ApproximatelyCT}, $k$-clique counting \cite{Eden2017OnAT}, and high-order graphlet estimation \cite{Chazelle2001ApproximatingTM, Czumaj2004EstimatingTW, Hassidim2009LocalGP, Onak2011ANS}. In this model, we can analyze the query cost and the (local) running time separately. 
\begin{figure}[t]
    \centering
    \includegraphics[width=0.25\textwidth]{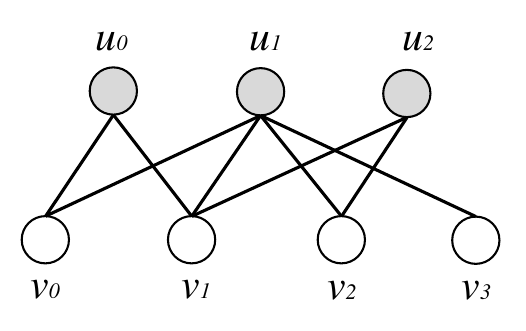}
    \vspace{-4mm}
    \caption{A bipartite graph instance.}
    \label{fig:example0}
    \vspace{-5mm} 
\end{figure}
\enlargethispage{\baselineskip}

\noindent \textit{Query Model.} 
Given a graph $G$, we consider algorithms that can sample an edge uniformly (with one query cost) and perform the following three types of queries. 
\begin{itemize}
    \item (Degree query): given a vertex $v$, it returns its degree $d_v$.
    \item (Neighbor query): given a vertex $v$ and an integer $i$, it returns the $i$-th neighbor of $v$;
    \item (Vertex-pair query): given a pair of vertices $u$ and $v$, it returns whether $(u, v)$ forms an edge in $G$.
\end{itemize}
The query model serves as a foundational framework for many graph parameter estimation tasks, and we provide a brief overview of related works in Section \ref{sec:related}.
In this paper, we study the approximate butterfly counting problem under the query model for the first time. Since query costs often dominate real-world execution time, minimizing the number of queries while maintaining accuracy is crucial.





\medskip \noindent
{\bf Existing studies.} 
{
As sampling is a useful strategy for approximate motif counting, state-of-the-art studies~\cite{sanei2018butterfly,zhang2023scalable} on approximate butterfly counting rely on various sampling techniques. We briefly introduce their approaches below and leave a detailed discussion to \Cref{sec:sota}, with a comparison provided in \Cref{sec:our} after we present our algorithm. We note that sketch-based methods are sometimes also used for approximate counting \cite{DBLP:journals/tkde/LiWJZZTYG22, DBLP:conf/focs/KallaugherKP18, DBLP:journals/pvldb/WangQSZTG17}; however, in static graph settings they offer limited advantages and are generally ill-suited for capturing butterfly substructures via simple mappings.

}

In~\cite{sanei2018butterfly}, a sampling-based algorithm \naive is proposed. It reads a constant fraction of the graph (around 20\%) and uses an exact butterfly counting algorithm~\cite{Wang2018VertexPB} on this subgraph to estimate the total number of butterflies. This method reduces the running time and achieves high empirical accuracy. However, it does not reduce the asymptotic dependence on the graph size. Specifically, its time complexity remains $O(mn)$—the same as the exact algorithm in~\cite{Wang2018VertexPB}. As a result, both the running time and query complexity remain prohibitive for large graphs.

In a very recent work~\cite{zhang2023scalable}, the authors propose a sampling-based algorithm \bs, which improves upon the algorithm in~\cite{sanei2018butterfly} by reducing the running time while maintaining good empirical accuracy. This work also provides a theoretical guarantee under a random graph model. However, the analysis does not extend to general graphs, and the random graph assumption may not hold for real-world data. Without the assumption, their analysis of the running time (also the query complexity) and estimation accuracy no longer holds. 
As a result, in some special graphs, the algorithm may incur high running time (also the query complexity) and fail to achieve good accuracy in practice.

{
Although butterfly counting has been studied in contexts such as streaming graphs \cite{Wang2022AcceleratedBC}, uncertain graphs \cite{Zhou2021ButterflyCO}, temporal graphs \cite{Cai2023EfficientTB}, distributed environments \cite{Weng2023DistributedAT}, and hierarchical memories \cite{Wang2023IOEfficientBC}, their focuses differ from ours. In the streaming setting, each edge insertion triggers updates to auxiliary data, incurring at least a linear factor in $|E|$. Consequently, these methods cannot attain the sublinear time or query complexity we target and often require storing a large portion of the graph \cite{Li2021ApproximatelyCB, SaneiMehri2018FLEETBE, Sheshbolouki2021sGrappBA}. In the distributed setting, prior work primarily focus on optimizing vertex-to-vertex communication and memory usage \cite{Weng2023DistributedAT}, rather than the global query complexity.

}

\medskip\noindent
{\bf Our solutions.} 
We design the algorithm using the following approaches. Firstly, we propose a two-level sampling algorithm \our, incorporating key ideas for variance control, such as ``small-degree first'' and ``constructing a wedge sampler based on a small subset''. This design allows us to reduce the query complexity to sublinear, regardless of the graph structure, thus improving the efficiency of \bs. With properly chosen parameters, experiments show that \our achieves good accuracy in practice, benefiting from the aforementioned variance control techniques.

Secondly, to further reduce variance in the worst case and provide robust theoretical guarantees, we incorporate two additional techniques: ``Heavy-light Partition'' and ``Guess-and-prove.'' After integrating these techniques into \our, we are able to establish strong guarantees on both accuracy and query complexity.




\medskip
\noindent
{\bf Contributions.} Our main contributions are as follows:
\begin{itemize}[leftmargin=*]
    \item We are the first to study efficient and accurate butterfly counting on bipartite graphs under the query model. It helps us characterize the real running time of the algorithm when query cost is much more expensive than local computation. Although our model considers only basic query operations while practical systems often process queries in batches, studying this basic model is fundamental for developing more advanced ones. 
    
    \item We propose a practical two-level sampling method \our estimating the number of butterflies without accessing a large portion of the graph, {leading to a sublinear query complexity}.
    
    \item We provide an explicit analysis of the accuracy, runtime, and query complexity of our algorithm on general bipartite graphs, offering the first theoretical guarantees for butterfly counting without assuming random graph models. Our main results are formally stated in \Cref{thm:edge}, which demonstrates that arbitrarily high accuracy can be achieved in sublinear time when the input graph is dense.
    
    \item We conduct extensive experiments on {15} real-world datasets to evaluate our proposed algorithm \our, outperforming existing solutions \naive and \bs by several orders of magnitude regarding the query cost and runtime. Compared to \naive, both \bs and our \our algorithm achieve reasonable accuracy with significantly improved runtime and query cost. 
    Compared to \bs, our \our algorithm achieves a comparable accuracy (i.e., with relative errors within 6\% across {most of the} datasets, compared to up to 10\% for \bs), while requiring substantially less query cost and running time. In the most notable case \texttt{edit-frwiki}, \our reduces query cost and running time by 1097$\times$ and 432$\times$, respectively. 
    {On average across all 15 datasets, \our improves over \bs by $179\times$ in query cost and 76$\times$ in running time.}
\end{itemize}



\section{Problem Definition}

In an undirected graph $G$ with $n$ vertices and $m$ edges, we use the notations shown in Table~\ref{tab:notation} to denote the parameters of $G$. Then we define a partial relation $\prec$ among the vertices in $V$ as follows. Firstly, we arbitrarily fix an order $\pi$ of vertices in $V$.
Then, for any two distinct vertices $u, v\in V$, if $d_u<d_v$ or $d_u=d_v$ and $u$ is before $v$ in order $\pi$, we set $u \prec v$.


\begin{definition}[Wedge]Consider a bipartite graph $G$, and vertices $u, v, w\in V$. A wedge $\wedge=(u, v, w)$ is a path starting from $u$, going through $v$, and ending at $w$.
\end{definition}

\begin{definition}[Butterfly]Consider a bipartite graph $G(U, L)$, vertices $u_1, u_2 \in U$ and $v_1, v_2 \in L$. A butterfly $\btf=\{u_1,u_2,v_1,v_2\}$ is  a complete bipartite subgraph (i.e., $2 \times 2$-biclique) induced by $u_1, u_2, v_1, v_2$.
\end{definition}

\begin{table}[htb]
\centering
\vspace{-2mm} 
\caption{Summary of Notations}
\begin{tabular}{c|c}
\noalign{\hrule height 1pt}
{\bf Notation} & {\bf Definition} \\ 
\noalign{\hrule height 0.6pt}
$G$ & a bipartite graph \\
$V,E$ & the vertex set of $G$, the edge set of $G$ \\
$n,m$, & the vertex number of $G$, the edge number of $G$ \\
$L,U$ & the lower vertex set of $G$, the upper vertex set of $G$ \\
$N(u)$ & the neighbors set of $u$ in $G$ \\
$d_u$ & the degree of vertex $u$ in $G$ \\
$d_e$ & the sum of its endpoints' degree minus two\\
$b$ & the number of butterflies in $G$ \\
$b(e)$ & the number of butterflies containing edge $e$ \\
$\btf,\wedge$ & a butterfly, a wedge \\
$w$ & the number of wedges in $G$\\ 
$W(e)$ & the set of wedges containing $e$ \\ 
$B(\wedge)$ & the set of butterflies containing $\wedge$ \\
\noalign{\hrule height 1pt} 
\end{tabular}
\label{tab:notation}
\end{table}

\noindent
{\bf Problem statement.} 
Given a bipartite graph $G$ with $n$ vertices and $m$ edges, and access to uniform edge sampler, we aim to \textbf{estimate the number of butterflies $b$ in $G$} under the query model that supports degree query, neighbor query, and vertex-pair query.


\medskip\noindent{\bf Remark.}
We adopt uniform edge sampling, supported by many real-world data centers \cite{DBLP:conf/soda/EdenR18,DBLP:journals/algorithmica/AliakbarpourBGP18}. When only vertex sampling is available, the technique in \cite{ENT2023} can simulate edge sampling with an additional query complexity of $\tilde{O}(n/\sqrt{m})$.

The paper is organized as follows: \Cref{sec:sota} reviews existing solutions, \Cref{sec:our} presents our core algorithm \our, \Cref{sec:theory} proposes two techniques for theoretical robustness, \Cref{sec:experiments} discusses experimental results and practical significance, and \Cref{sec:related} reviews related work on butterfly counting and query models and provide the implications of our work.

\noindent
\section{Existing Solutions}
\label{sec:sota}

In this section, we briefly discuss the two state-of-the-art algorithms: Edge Sparsification (\naive)~\cite{sanei2018butterfly} and Weighted Pair Sampling (\bs)~\cite{zhang2023scalable}.
\subsection{\naive}

The \naive method samples a fraction of edges with probability $p$ to form an induced subgraph $G'$, counts butterflies $\btf(G')$ exactly, and estimates the butterfly count for the original graph $G$. Compared to the exact algorithm, \naive reduces running time while maintaining high accuracy, but with a non-negligible $p$ according to \cite{sanei2018butterfly}, its query complexity remains $O(mn)$, matching the exact algorithm’s growth trend.



Consequently, as shown in Section~\ref{sec:res}, although showing the best performance in relative error on almost every graph, both its running time and number of queries are significantly higher than those of our method.

\begin{lemma}[Peak Memory Usage of The \naive Algorithm]  \label{lm:memory_of_espar}
{The peak memory usage of Algorithm~\ref{Alg_espar} is $O(p\cdot|E| + |V|)$. } 
\end{lemma}
\begin{proof}
    { 
    The dominant term of the memory consumption comes from storing a fraction of the edges of the input graph, i.e., $O(p\cdot|E| + |V|)$. Here $p$ is recommended to be set as a constant independent of the size of the graph according to~\cite{sanei2018butterfly}.
}
\end{proof}

\vspace{-4mm} 
\begin{algorithm}[h]
\caption{Edge Sparsification  (\naive)}\label{Alg_espar}
\KwInput{Graph $G$ and a probability parameter $p$}
\KwOutput{number of butterflies}
Sample an edge set $E'$ uniformly at random with probability $p$ and denote by $G'$ the induced subgraph\;
Let \btf($G'$) be the number of butterflies of the induced subgraph $G'$\;
\Return{$(\btf(G')/4)\times p^{-4}$}\;
\end{algorithm}
\vspace{-5mm} 

\subsection{\bs} 
\bs first performs degree queries on one vertex layer $U$ (or $L$) to get the number of edges $m$. It then samples two vertices $u, v$ independently with probabilities proportional to their degrees, i.e., $d_u/m$ and $d_v/m$, respectively. Finally, it estimates the total number of butterflies based on the number of butterflies incident to $u$ and $v$.

\begin{lemma}[Peak Memory Usage of The \bs Algorithm]
\label{lm:memory_of_wps}
{The peak memory usage of Algorithm~\ref{Alg_wps} is $O(\min(|L|,|U|))$. } 
\end{lemma}
\begin{proof}
    {
    
    \bs does not store the edges of the input graph.  
    Degree-based vertex sampling (line~\ref{need_store}) requires storing the degrees of all vertices in a layer, with memory $O(S) = O(\min(|L|,|U|))$.  
    Other computations use only a constant number of temporary variables.  
    Hence, the peak memory usage is $O(\min(|L|,|U|))$.
}
\end{proof}

\vspace{-3mm} 
\begin{algorithm}[h]
\caption{Weighted Pair Sampling (\bs)}\label{Alg_wps}
\KwInput{Graph $G$}
\KwOutput{number of butterflies}
Let $S \gets U$ or $L$\;
Sample a vertex $u\in S$ with probability $d_u/m$, and a vertex $v\in S$ with probability $d_v/m$\label{need_store}\;
Compute the size of intersection of the neighbors of $u$ and $v$, i.e., $N_{u,v}=|N_v\cap N_u|$\;
\If{$u = v$}{
        \Return{$0$}\;
    }
$\btf_{u,v}\leftarrow\tbinom{N_{u,v}}{2}$\;
\Return{$\frac{m^2}{2d_ud_v}\cdot\btf_{u,v}$}\;
\end{algorithm}
\vspace{-2mm} 

Although \bs reduces queries compared to \naive, it remains costly in some cases, as shown in \Cref{sec:res}. This is mainly because pair sampling can be quite expensive when vertices $u$ and $v$ both have high degrees, requiring a large number of queries to count their common neighbors.
In addition, since high-degree vertex pairs are sampled with higher probability, such costly computations occur frequently. 
Consequently, the query performance of \bs is significantly affected by the graph structure (e.g., the distribution of edges, variance of degrees, etc).


\begin{figure}[h]
    \vspace{-5mm} 
    \centering
    \includegraphics[width=7.8cm]{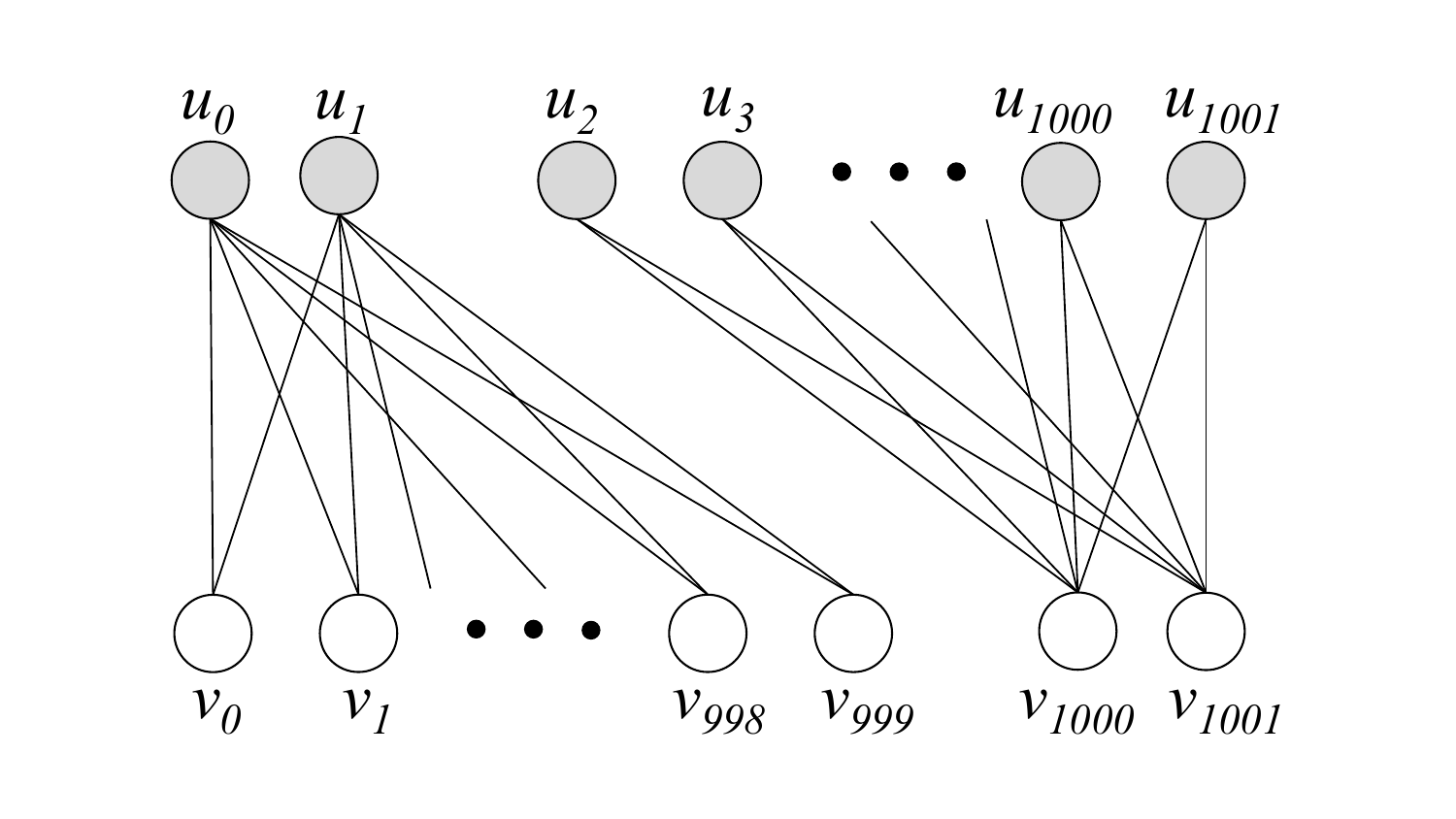}
    \vspace{-6mm} 
    \caption{A bipartite graph containing high degree vertices $u_0$, $u_1$, $v_{1000}$ and $v_{1001}$.}
    \vspace{-2mm} 
    \label{fig:bad_example}
\end{figure}

Consider the bipartite graph in Figure~\ref{fig:bad_example} as an example.

1) For running time and query complexity: \bs typically queries common neighbors of high-degree vertices (e.g., $u_0$, $u_1$), requiring traversal of large graph portions. For example, checking whether all neighbors of $u_0$ ($v_0, v_1, \ldots, v_{999}$) connect to $u_1$ takes $O(d_{u_0})$ queries per wedge (e.g., 1000 here). Since sampled vertices are typically high-degree, the query cost can scale with graph size.

2) For estimation accuracy: In this graph, a single sampling round of \bs can has a relatively high probability of encountering the following case with notable estimation errors, leading to large variance. With probability about $\tfrac{1}{2}$, it selects one high-degree vertex (e.g., $u_0$ or $u_1$) and one low-degree vertex (e.g., one from $u_2$–$u_{1001}$). Their neighbor intersection is 0, and returns the estimate:


\begin{equation*}
\frac{m^2}{2d_ud_v} \cdot \binom{0}{2} = \frac{4000 \times 4000}{2 \times 1000 \times 1000} \cdot \binom{0}{2} = 0
\end{equation*}

However, the actual number of butterflies is $2 \times \binom{1000}{2}$, resulting in a significant relative error.

Recognizing the limitations of existing approaches for robust butterfly counting approximations, we propose our two-level sampling algorithm \our in the following section.

For 1), our \our algorithm is developed from a more local perspective and incorporates the sample-based approach, effectively reducing the query cost per wedge.

For 2), unlike the vertex-pair-based sampling method \bs, our method uniformly samples edges across the entire graph in its first-level sampling, avoiding the frequent occurrence of $0$ in one sampling round. 

A more detailed comparison based on the same example will be discussed in Section \ref{sec:our}.



\section{The Two-level Sampling Algorithm}
\label{sec:our}
\subsection{Overview}

Although existing approaches incorporate sampling techniques to reduce query cost, they often access a large portion of the graph (up to a constant fraction), resulting in high query complexity and preventing sublinear performance in many cases. To overcome this limitation, we introduce our two-level sampling algorithm \our, which is implemented by a nested loop with two levels.

In each iteration $i$ of the outer loop, we uniformly sample a fixed-size edge set $S_i$, which serves as a 
representative subset of the graph. 
Thereafter, within the inner loop, which runs for $s_2$ rounds, we estimate the number of butterflies incident to $S_i$.
Specifically, in each inner round $j$, we first sample a wedge $\wedge_j$ uniformly at random from the $S_i$-related wedges (i.e., the wedges containing at least one edge in $S_i$).
Then we estimate the number of butterflies incident to the sampled wedge $\wedge_j$ by sampling $R$ additional neighbors of the wedge uniformly at random and testing whether it forms a butterfly with $\wedge_j$.
Through a scaling factor, we can finally estimate the total number of butterflies in the graph.

Compared to \bs, which estimates the butterfly count by sampling vertex pairs and enumerates all butterflies between them, our method focuses on wedges and retains a fully sampling-based estimation procedure. By limiting the number of neighbor queries per wedge to $R$, we ensure the query complexity remains controlled regardless of graph structures.

To further control the variance, we also incorporate several variance-reduction techniques and carefully tune algorithmic parameters. For instance, we adopt a small-degree neighbor sampling strategy when checking for butterfly completion. These techniques are essential in both theoretical analysis and practical implementation to achieve reliable estimates with low variance, while maintaining fast runtime and low query cost.

We show the details on the implementation of \our in Algorithm \ref{alg:toy_our_main_algo}.
The algorithm performs $r$ rounds of outer sampling. In the $i$-th round, it uniformly samples an edge set $S_i$ of size $s_1$ from the original graph edge set. 
Then, in the inner sampling loop, it estimates the global butterfly count $b(S_i)$. 
Finally, the algorithm returns the average of the $r$ estimates (i.e., $\sum_{i=1}^{r} b(S_i)/r$) as the estimated butterfly number for the entire graph $G$.

Before the inner sampling, we first construct a uniform wedge sampler inside $S_i$.
We perform degree queries to obtain the degrees $d_u$ and $d_v$ for both endpoints of each edge $e=(u,v) \in S_i$, and compute the ``edge degree" $d_e$, which enables us to construct a uniform wedge sampler. 
We then sum all $d_e$ for $e$ in $S_i$, denoted by $W(S_i)$, which is also the total number of potential wedges that could be sampled.

Next, we perform $s_2$ rounds of inner sampling to estimate the number of butterflies incident to $S_i$, $b(S_i)$.
In each round $j$, we uniformly sample a wedge $\wedge_j$ and estimate the number of butterflies containing this wedge, denoted as $\hat{b}(\wedge_j)$, through the following procedure:
\begin{enumerate}[leftmargin=*]
    \item Sample an edge $e= (u,v)$ from $S_i$ with probability proportional to its degree $d_e$;
    \item Select either $u$ or $v$ as the middle vertex of the wedge with probability proportional to their degrees ($d_u$ and $d_v$ respectively).
    Without loss of generality, we assume $u$ is selected.
    Using neighbor queries, we randomly select a neighbor $x_j$ of $u$, forming the sampled $\wedge_j = (v, u, x_j)$ composed of edges $e$ and $(u, x_j)$.
    \item We then estimate the butterfly count $b(\wedge_j)$ containing $\wedge_j$ using a method of \emph{degree-first sampling}, where we try to find a butterfly from the smaller endpoint of the wedge. 
    Specifically, let $y_j$ be the vertex with smaller degree between $v$ and $x_j$ in the $\wedge_j$. 
    We then determine the sampling times $R$ based on $d_{y_j}$. 
    Once $R$ is determined, we proceed to sample $R$ additional neighbors of $y_j$ uniformly at random and verify whether it closes the wedge into a butterfly.
    If it ends up forming a butterfly, we increment our estimate by $d_{y_j}/4$ (the factor $4$ accounts for each butterfly appearing in $4$ distinct wedges). 
    The final estimate $\hat{b}(\wedge_j)$ averages results from all $R$ trials.
\end{enumerate}

\subsection{Implementation and Efficiency}
\begin{algorithm}[htbp]
\caption{The two-level sampling algorithm \our}
\label{alg:toy_our_main_algo}
\KwInput{a graph $G=(V,E)$}
\KwOutput{estimate of butterfly count in $G$}

\SetKwProg{Fn}{Function}{:}{}
\For{$i=1, \ldots, r$}{
\codecmt{Representative edge set}

Sample a set of edges $S_i$ with size $s_1$\label{line:sample_repre_set}\;
Query the degrees of $u$ and $v$ for all $e=(u,v) \in S_i$ and let $d_e = d_u + d_v - 2$\label{line:sample_deg_set}\;
\For{$j=1,\ldots, s_2$}{
\label{line:inner_begin}\codecmt{Sample a wedge $\wedge_j$}

Sample an edge $e= (u, v)$ from $S_i$ with probability proportional to $d_e$\label{line:sample_edge}\;
Choose $u$ with probability $(d_u-1)/d_e$ and $v$ with probability $(d_v-1)/d_e$\label{line:choose_u}\;
Sample a neighbor $x_j$ of the selected vertex (denoted by $u$ w.l.o.g)\;
Let $\wedge_j$ be the wedge $(v, u ,x_j)$\label{line:get_wedge}\;

\codecmt{Estimate $b(\wedge_j)$}

Let $y_j \leftarrow \mathrm{argmin}_{v, x_j} \{d_v, d_{x_j}\}$\;
Let $R = \max(10 \times d_{y_i}/\sqrt{m},10)$\;
\For{$k=1,\dots,R$}{ 
    $Z_k \leftarrow 0$ and sample a vertex $z_k\in N(y_j)$\;
    \If{$z_k$ and $\wedge_j$ is a butterfly and $x_j \prec z_k$}
    {
        $Z_k\leftarrow d_{y_j}/4$ \;
    }
}
Let $\hat{b}(\wedge_j) \leftarrow \frac{1}{R}\sum_{k=1}^R Z_k$ \label{line:b_wedge_j}\;
}
Let $\hat{b}(S_i) \leftarrow \frac{m}{s_1 s_2}\cdot \sum_{(u,v)\in S_i}d_{(u,v)}\cdot\sum_{j=1}^{s_2} \hat{b}(\wedge_j)$\label{line:inner_end}\;
}
\Return{$\frac{1}{r}\cdot \sum_{i=1}^{r} \hat{b}(S_i)$}\;
\end{algorithm}
\vspace{-3mm}

After obtaining $\hat{b}(\wedge_j)$ for each sampled wedge, we compute the average butterfly count per wedge from $s_2$ samples. The global estimate from $S_i$ is then:
\begin{equation*}
    \hat{b}(S_i) = \left(\frac{1}{s_2}\sum_{j=1}^{s_2} \hat{b}(\wedge_j)\right) \times W(S_i) \times \left(\frac{m}{s_1}\right) 
\end{equation*}
This can be interpreted as: butterflies per wedge $\times$ wedge count in $S_i$ $\times$ scaling factor to the full graph. The algorithm returns $b(S_i)$ as the outer-layer estimate for the $i$-th round.

\begin{lemma}[Efficiency of The \our Algorithm]
Under a uniform edge sampler, the running time and the query complexity of Algorithm~\ref{alg:toy_our_main_algo} are $O(r \cdot (s_1+s_2\cdot R))$.  
\end{lemma}
\begin{proof}
    The outside loop is $r$. The complexity of line~\ref{line:sample_repre_set} is $O(s_1)$. and the complexity of the inside loop is $R$. In conclusion, we have $O(r \cdot (s_1+s_2\cdot R))$.
\end{proof}

\begin{lemma}[Peak Memory Usage of The \our Algorithm]
{ Under a uniform edge sampler, the peak memory usage of Algorithm~\ref{alg:toy_our_main_algo} is $O(s_1)$. According to our settings in \cref{sec:experiments}, $s_1$ is bounded by $0.5 \sqrt{|E|}$. Since real-world graphs are typically sparse, the memory cost of $O(\sqrt{|E|})$ is often much lower than that of memory growing linearly with the number of vertices.

}  
\end{lemma}
\begin{proof}
    {
    Similar to \bs, \our does not store a subgraph of the input graph.  
    The outer sampling stores the $S_1$ sampled edges (line~\ref{line:sample_repre_set}) and the degrees of their endpoints (line~\ref{line:sample_deg_set}), using $O(S_1)$ space.  
    All other computations, including the inner sampling (lines~\ref{line:inner_begin}--\ref{line:inner_end}), use only a constant number of temporary variables.  
    Thus, peak memory is $O(S_1)$.
}
\end{proof}

\vspace{-2mm}
{ After showing a originally analysis in \Cref{lm:memory_of_espar} and \Cref{lm:memory_of_wps}, we evaluated the actual peak memory usage of \naive, \bs, and \our during execution and reported the results in \Cref{table:peak_memory} of \Cref{sec:experiments}. Across the 15 datasets, \our consistently exhibits the lowest peak memory.



}

Finally, we discuss the choice of parameter settings. Different configurations are used in the theoretical analysis and in the experiments, for two main reasons. First, the theoretical guarantees are designed for worst-case instances, which often do not represent typical real-world graphs. Second, the analysis hides large constants in the big-$O$ notation, which can dominate performance when the input size is moderate. Despite these limitations, the theoretical bounds remain valuable: they ensure robustness in near-worst-case scenarios or when the graph is extremely large.

For practical purposes, we adopt an automatic parameter configuration for $r$ and $s_2$, as detailed in \Cref{sec:experiments}. Briefly, we terminate the sampling loops when additional rounds lead to negligible changes in the current estimate. For $s_1$, we set $s_1 = 0.5 \times \sqrt{m}$ as a balanced choice between efficiency and estimation variance.

For the theoretical setup, in order to guarantee a $(1+\varepsilon)$-approximation, the parameters must be set as a function of $\varepsilon$. In addition, we introduce several theoretical techniques, such as the heavy-light partition and the guess-and-prove framework, which will be presented in the next section.

\begin{example}
Consider the bipartite graph $G$ in Figure~\ref{fig:bad_example}. We provide an example to compare our method with \bs. 

At the first sampling level, \bs selects a high-degree and a low-degree vertex with probability about $\frac{1}{2}$, leading to large estimation errors (Section~\ref{sec:sota}). For such graphs, \bs thus requires very large samples to reach acceptable accuracy, incurring high query costs. In contrast, \our samples edges uniformly from $S_i$, giving each edge equal probability. This prevents the \bs issue of sampling vertex pairs with no common neighbors (yielding $0$ estimates) and improves robustness.


At the second sampling level, as we discussed in Section \ref{sec:sota}, \bs queries the common neighbors of high-degree vertices with high probability, incurs up to 1000 queries per wedge in this example. Our method \our addresses this inefficiency by querying only a fixed number $R$ of partial neighbors for each sampled wedge (e.g., $\wedge(u_0, v_0, u_1)$). According to Algorithm~\ref{alg:toy_our_main_algo}, $R$ is set to 158 for this example and typically reduces to 10 in practice, drastically lowering query overhead.
\end{example}

\section{Theoretical Gurantees}
\label{sec:theory}

In this section, we introduce the two techniques \emph{heavy-light} and \emph{guess-and-prove} under \our to ensure the aforementioned theoretical guarantees, inspired by \cite{eden2017approximately}. 
First of all, we present our main theoretical results as follows.
\begin{theorem} 
\label{thm:edge}
    Under a uniform edge sampler, given an approximation parameter $0<\epsilon<1$, our algorithm can output an estimate $\hat{b}$ of $b$ satisfying $(1-\epsilon)b<\hat{b}<(1+\epsilon)b$ with high constant probability. The running time and the query time are all bounded in 
    $$
    O\left(\left(w\sqrt{m}/{b} + m/b^{\frac14}\right)\cdot \mathrm{poly}(\log n,\epsilon)\right)\,.
    $$
\end{theorem}

We give some examples to illustrate the implications of our result.
Consider a dense graph where $w = \Theta(n^3)$, $m = \Theta(n^2)$, and $b = \Theta(n^4)$.
To ensure a constant relative error (i.e., $\epsilon$ is constant), our algorithm only needs to run in $O(\mathrm{poly}(\log n))$ running time and edge queries.
Further, in a more realistic scenario where the edges are more sparse, $b = \Theta(n^3)$, $m = \Theta(n^{5/4})$, and $w = \Theta(n^3)$, the running time and the number of edge queries are still bounded by $O(\sqrt{n} \cdot \mathrm{poly} (\log n))$, which remains sublinear.
Below, we provide overviews of the two techniques used in our algorithm.

\smallskip\noindent
{\bf Technique 1: heavy-light.}
We first propose \emph{heavy} and \emph{light} edges, extending the heavy and light vertices of \cite{eden2017approximately}. Intuitively, heavy edges are involved in many butterflies, while light edges are not. If we do not distinguish the two types of edges, the butterfly number of some edges could potentially be very large, which causes a large variance to $\hat{b}(\wedge_j)$ (at line~\ref{line:b_wedge_j} of \our). We show in \Cref{prop:num_of_btf_with_nonlight_edges} that butterflies with only heavy edges occupy a relatively small fraction, allowing minimal loss when discarding them.

Note that our definitions of heavy and light edges rely on the concept of the number of butterflies, which has not been determined yet.
We observe that, rather than using the exact values, it suffices to use a rough estimation $\bar{b}$ of $b$ that satisfies a weakened requirement (Assumption~\ref{aspt:b_m_bound}). 
This leads to the following technique: guess-and-prove.  
    
\smallskip\noindent
{\bf {Technique 2: guess-and-prove.}} 
We use geometry search to find a valid $\bar{b}$ satisfying Assumption~\ref{aspt:b_m_bound}.
Each time, we guess a value $\bar{b}$ and use it as an estimation of $b$. Then, we move to the proof phase. Roughly speaking, our proof phase will accept the guess when it is accurate enough and return a better estimation of $b$. On the other hand, it will reject the estimation and move to the next guess phase if it is not accurate enough. 
    
\subsection{Heavy-light Partition}
\label{sec:partition_of_edge_set}
We first introduce the definitions of heavy and light edges, which are defined on two estimates $\bar{w}$, $\bar{b}$, and the approximation parameter $\epsilon$. 
The following assumption provides the requirement for the estimation.

\begin{assumption}{\label{aspt:b_m_bound}
The approximations of the number of butterflies and wedges, $\bar{b}$, $\bar{w}$, satisfy the following inequalities,
$$
{b}/2 \le \bar{b} \le  2b,\quad
 w/6\le \bar{w} \le 6w\,.
$$}
\end{assumption}

\begin{definition}[Heavy and light edge]\label{def:heavy_light}
Given $\bar{b}, \bar{w}, \epsilon$ and an edge $e$, we say $e$ is {\bfseries heavy} with respect to $\bar{b}, \bar{w}, \epsilon$ if either $b(e)>2\bar{b}^{\frac{3}{4}}/{\epsilon^\frac{1}{4}}$ or $d_e>\bar{w}/(\epsilon \bar{b})^\frac{1}{4}$ holds. 
If $b(e)<\bar{b}^{\frac{3}{4}}/{2\epsilon^\frac{1}{4}}$ and $d_e<\bar{w}/(\epsilon \bar{b})^\frac{1}{4}$, $e$ is said to be {\bfseries light} with respect to $\bar{b},\bar{w}, \epsilon$.
\end{definition}
Intuitively, an edge is labeled heavy if it is involved in a sufficiently large number of butterflies or its endpoints have large degrees in total, and an edge is light if both conditions (precisely, of weaker versions) are violated.
It is worth noting that there is a gap between heavy and light edges, and we allow an edge to be neither heavy nor light. 

Next, \Cref{prop:num_of_btf_with_nonlight_edges} proves that the number of butterflies composed solely by non-light edges is relatively small compared to the total number $b$, which provides us the key insight.

\begin{restatable}{proposition}{PropNonlightEdges}
\label{prop:num_of_btf_with_nonlight_edges}
When Assumption~\ref{aspt:b_m_bound} holds, at most $c_H\epsilon b$ butterflies solely consist of non-light edges with $c_H= 1.77\times 10^4$.
\end{restatable}
\begin{proof}
Observe that $\sum_{e\in E} b(e) = 4\cdot b$ since each butterfly is counted four times in the summation.
Hence, the number of edges such that $b(e)\ge \bar{b}^{\frac{3}{4}}/{2\epsilon^\frac{1}{4}}$ is no more than
\[
\frac{4\cdot b}{\bar{b}^{\frac{3}{4}}/{2\epsilon^\frac{1}{4}}} \le
\frac{4\cdot b}{(b/2)^{\frac{3}{4}}/{2\epsilon^\frac{1}{4}}} = 8\sqrt[4]{8}
\cdot (\epsilon b)^\frac{1}{4} < 13.5\cdot  (\epsilon b)^\frac{1}{4}.
\]

Similarly, $\sum_{e\in E} d_e  =  2\cdot w$ since each wedge is counted twice.
Thus, the number of edges such that $d_e \ge w/(\epsilon \bar{b})^\frac{1}{4}$ is at most
\[
\frac{2w}{(\bar{w})/(\epsilon \bar{b})^\frac{1}{4}} \le
12\cdot(\epsilon b)^{\frac{1}{4}}.
\]
Summing them up, the number of non-light edges is at most $25.5\cdot (\epsilon b)^{\frac14}$.
Therefore, the following number of butterflies formed by these edges is bounded by  
\begin{equation*}
\binom{25.5\cdot(\epsilon b)^{\frac{1}{4}}}{4}< \frac{(25.5\cdot(\epsilon b)^{\frac14})^4}{4!} <c_H\epsilon b \,. \qedhere
\end{equation*}
\end{proof}


Based on 
on the above proposition
, we still maintain a $(1-c_H\epsilon)$ approximation after discarding the butterflies formed only by non-light edges.

Next, we present a subroutine \heavy to classify an edge as heavy or light. Since $b(e)$ can be large for certain edge $e$, an exact algorithm determine whether an edge is heavy or light would be slow. Instead, we propose a stochastic algorithm ${\tt Heavy}(e, \bar{b}, \bar{w}, \epsilon, m)$ (Algorithm~\ref{alg:heavy}) that determines whether $e$ is heavy by sampling vertices to estimate $b_e$.
The correctness is guaranteed by the following theorem, with proof deferred to Appendix IX.A of our {technical} report \cite{fullversion}.

\begin{restatable}[Correctness of \heavy]{lemma}{lemCorrectHeavy}
\label{thm:correctness_of_heavy}
For any edge $e$ labelled heavy or light, ${\tt Heavy}(e, \bar{b}, \bar{w}, \epsilon, m)$ returns the correct answer with a probability at least $1-1/{m^2}$.        
\end{restatable}

For convenience, we denote $O(f \cdot \text{poly}(\log n, \epsilon))$ as $O^*(f)$.
The expected time cost of ${\tt Heavy}(e, \bar{b}, \bar{w}, \epsilon, m)$ is given by the following lemma.

\begin{restatable}[Efficiency of \heavy]{lemma}{lemEffHeavy}
\label{lem:running_time_of_heavy}
If the second condition of Assumption~\ref{aspt:b_m_bound} holds, then the expected query time and the running time of ${\tt Heavy}(e, \bar{b}, \bar{w}, \epsilon, m)$ is $O^*(\bar{w}\sqrt{m}/{\bar{b}})$.
\end{restatable}

\begin{proof}
First, we show that the expected running time of an iteration starting from Line~\ref{line:start_of_iteration_j} to Line~\ref{line:end_of_iteration_j} is $O(1)$.
As the edge $e$ is given, the time of sampling a wedge containing $e$ only costs $O(1)$ time.
If $d_e \le \sqrt{m}$, then the degree of $y_j$ will also be no more than the maximum of $d_u$ and $d_v$, which is at most $\sqrt{m}$.
Thus, $r$ will be assigned zero, and the running time is clearly $O(1)$.
Otherwise, let $s\in \{u, v\}$ be the selected vertex at Line~\ref{line:choose_u}.
Then the running time $R$ will be no more than $1 + \frac{\min(d_s, d_{x_j})}{\sqrt{m}}$.
Thus, the upper bound of the expected time cost of Line~\ref{line:start_of_iteration_j} to Line~\ref{line:end_of_iteration_j} is given by 
\begin{align*}
& \sum_{s\in \{u, v\}} \sum_{x_j\in N(s)} \frac{d_s - 1}{d_e}\cdot \frac{1}{d_s - 1}\cdot O\left( \frac{\min(d_s, d_{x_j})}{\sqrt{m}}\right) \\
& \le \sum_{s\in \{u, v\}} \sum_{x_j\in N(s)} \frac{d_s - 1}{d_e}\cdot \frac{1}{d_s - 1}\cdot O\left( \frac{d_{x_j}}{\sqrt{m}}\right) \\
& = \frac{1}{d_e\cdot \sqrt{m}}\cdot  \sum_{s\in \{u, v\}} \sum_{x_j\in N(s)} O(d_{x_j})  \\
& \le \frac{1}{d_e\cdot \sqrt{m}}\cdot O(4m) \tag{every edge counted at most 4 times}\\
& = O(1)  \tag{as $d_e \ge \sqrt{m}$}
\end{align*}

Next, as the above loop is repeated $O(\sqrt{m}\cdot \bar{w}/(\epsilon^2\bar{b}) $ at Line~\ref{line:start_of_iteration_j} and $i$ traverse from $1$ to $48\log(2m)$, the expected time cost is given by $O(\log(m) \cdot \sqrt{m}\cdot \bar{w}/(\epsilon^2\bar{b})$. 
Since $O(\log m) = O(\log n)$, the expectation can also be written as $O(\bar{w}\cdot \sqrt{m}/\bar{b}\cdot \mathrm{poly}(\log n, 1/\epsilon^2)) = O^*(\bar{w}\cdot \sqrt{m}/\bar{b})$.
\end{proof}


Applying the \heavy subroutine then yields an edge partition $(P_H, P_L)$, where edges labelled ``heavy'' are included in $P_H$ and other edges are in $P_L$.
We say the partition is \emph{appropriate} if every heavy or light edge is grouped correctly with high probability, {while we allow an unlabeled edge to be classified arbitrarily.}
By \Cref{thm:correctness_of_heavy} and the union bound, the partition $(P_H, P_L)$ is appropriate with high probability. 
\begin{corollary}\label{coro:prop_of_appropriate_partition}
Applying \heavy for all edges yields an appropriate partition $(P_H, P_L)$ with probability of at least $1-1/m$.
\end{corollary}

Note that although our theoretical analysis below is built upon the edge partition, our algorithm (see Algorithm~\ref{alg:estimate_s1_s2}) does not apply \heavy to all edges in advance.
Instead, we assume that there is an underlying edge partition $(P_H, P_L)$ -- as if \heavy were applied to all edges -- and each time we invoke \heavy to query whether an edge is in $P_H$ or $P_L$.

\subsection{Butterfly Estimation}
\label{sec:btf_estimate}
In this part, we implement \our using the heavy-light technique.
As mentioned above, the partition induced by \heavy is appropriate with high probability.
We next define a weight function $\mathrm{wt}_{P_L}(\cdot)$ as the objective of approximation.
Once the partition is appropriate, we show that the sum of weights of edges in $P_L$ is quite close to the number of butterflies, and our algorithm approximates the weight sum well.

The basic idea of $\mathrm{wt}_{P_L}(\cdot)$ is to assign a weight of $1$ to every butterfly containing at least one light edge (precisely, an edge in $P_L$) and equally distribute the weight to every edge in $P_L$ contained in this butterfly.
Formally,


\begin{algorithm}[t]
\caption{Heavy-light (\heavy)}\label{alg:heavy}
\KwInput{an edge $e=(u,v)$, $\bar{b}$, $\bar{w}$, $\epsilon$, $m$}
\KwOutput{whether $e$ is heavy or light}
\SetKwProg{Fn}{Function}{:}{}
\If{$\bar{w} <(\epsilon \bar{b})^\frac{1}{4}\cdot d_e$}{
    \Return{heavy}\;
}
\For{$i=1$ to $t=48\log (2m)$}{\label{line:start_of_iteration_i}
    \For{$j=1$ to $s= 12\sqrt{m}\cdot \bar{w}/(\epsilon^2\bar{b})$} {
    \label{line:start_of_iteration_j}

        Perform the same procedure from line~\ref{line:choose_u} to line~\ref{line:get_wedge} in Algorithm~\ref{alg:toy_our_main_algo} to sample a wedge $\wedge_j=(v, u, x_j)$ containing $e$\;
        Define $y_j$ the same as in Algorithm~\ref{alg:toy_our_main_algo}\;
        
        Set $R \leftarrow \lceil d_{y_j}/\sqrt{m}\rceil$\;
        \For{$k=1, \ldots, R$}{
            Sample $z\in N(y_j)$ uniformly\;
            If $\wedge_j$ and $z$ form a butterfly and $x_j\prec z$ set $Z_k\leftarrow d_{y_j}$, else $Z_k\leftarrow 0$ \label{line:wz_form_btf}\;
        }
                $Y_j\leftarrow \frac1R \sum_{k=1}^R Z_k$ \label{line:end_of_iteration_j}\;
            }
        $X_i\leftarrow \frac1s \cdot \sum_{j=1}^s Y_j$\label{line:heavy_xi}\;
	}
Let $X$ be the median of $X_i$\; 
\Return{heavy \text{if }$X > \bar{b}^{\frac{3}{4}}/\epsilon^\frac{1}{4}$ \text{and} light \text{otherwise}}\;
\end{algorithm}
\begin{definition}\label{def:weight_btf}
Given a butterfly $\btf$ and a non-empty edge set $P_L \subseteq E$, we define the weight for each edge of the butterfly $\btf$ (denoted by $\mathrm{wt}_{P_L}(\btf)$) and the weight of an edge $e\in E$ (denoted by $\mathrm{wt}_{P_L}(e)$) as follows,
\begin{align*}
\mathrm{wt}_{P_L}(\btf) &=
\begin{cases}
0 & \text{if no edges of $\btf$ belong to ${P_L}$} \\
\frac{1}{\ell} & \text{if $\btf$ has $\ell>0$ edges that belong to ${P_L}$},
\end{cases} \\
\mathrm{wt}_{P_L}(e) &= 0, \quad \text{if $e\notin {P_L}$,} \\
\mathrm{wt}_{P_L}(e) &= \sum\nolimits_{e\in \btf} \mathrm{wt}_{P_L}(\btf), \quad \text{if $e\in {P_L}$\,.} 
\end{align*}
\end{definition}

It is clear that the sum of the weights assigned to the edges is equal to the number of butterflies with at least one edge in ${P_L}$.
By \Cref{prop:num_of_btf_with_nonlight_edges}, there is at most $c_H\epsilon$ fraction of butterflies solely consisting of non-light edges.
When the partition $({P_H}, {P_L})$ is appropriate, every light edge is included in ${P_L}$.
Therefore, at least $(1-c_H\epsilon)\cdot b$ butterflies intersect with the edge set ${P_L}$, which immediately implies the following lemma.
\begin{lemma} \label{lm:wt_upper_bound}
For any partition $({P_H}, {P_L})$, we have $\sum\nolimits_{e\in {P_L}} \mathrm{wt}_{P_L}(e) \le b$.
When $({P_H},{P_L})$ is appropriate and Assumption~\ref{aspt:b_m_bound} holds, then 
$$
b(1-c_H\epsilon)\le \sum\nolimits_{e\in {P_L}}\mathrm{wt}_{P_L}(e) \le b\,.$$
\end{lemma}

The idea of introducing the weight function is to estimate the number of butterflies by edge sampling:
sample a representative edge set $S$ (line~\ref{line:sample_repre_set} of \our) and estimate $b$ using the average weight $\frac{m}{\abs{S}}\sum_{e \in S} \mathrm{wt}_{P_L}(e)$.
We say $S$ is \emph{good} if the average weight lies within the interval $[b(1-2c_H\epsilon), b]$.
To identify a good set $S$ (i.e., one whose weighted sum is close to $b$), we present the following lemma:

\begin{restatable}{lemma}{lemWt}
\label{th:wt_bound_in_S}
If Assumption~\ref{aspt:b_m_bound} holds, for any appropriate edge partition $({P_H}, {P_L})$, suppose a uniformly sampled edge set $S$ has size of at least $s=c\cdot m\log(n/\epsilon^2)/(\bar{b}^{\frac{1}{4}} \cdot \epsilon^{\frac94})$.
Let variable $Y = \frac1s \sum_{e_i\in S} \mathrm{wt}_{P_L}(e_i)$. 
Then we have \[
b(1-c_H\cdot\epsilon) \le m\cdot \E\left[Y\right]\le b\,.
\]
When the constant $c$ is large enough (i.e., $c > 1/ (2 \cdot (1-c_H\cdot \epsilon))$), we have
\[\Pr\left[mY< b(1-2c_H\cdot\epsilon)\right]< \frac{\epsilon^2}{n}\,.\]
\end{restatable}

\begin{algorithm}[t]
\caption{Butterfly-Estimate-with-Guess (\btfe)}
\label{alg:estimate_s1_s2}
\SetKwProg{Fn}{Function}{:}{}

\KwInput{$\epsilon, m$, and two guessed values $\bar{b}, \bar{w}$}
\KwOutput{the estimated number of butterflies}


Sample an edge set $S$ with size $s_1$ as \Cref{th:wt_bound_in_S}\;


Let $s_2\leftarrow 40 (1+2c_H\cdot \epsilon)\bar{w} \sqrt{m} \log^2 n/(\epsilon^4\bar{b})$\;
\For{$i=1,\ldots, s_2$}{
Let $Y_i \leftarrow 0$\;
Perform the same procedure from line~\ref{line:sample_edge} to line~\ref{line:b_wedge_j} in \our to sample an edge $e=(u,v)$ a wedge $\wedge_i = (v, u, x_i)$ containing $e$;

Define $y_i$ the same as in Algorithm~\ref{alg:toy_our_main_algo}\;
\If{$d_{y_i}\le \sqrt{m}$}{
Set $R\leftarrow 1$ with probability $\frac{d_{y_i}}{\sqrt{m}}$ and $0$ otherwise\;
}\Else{
Set $R\leftarrow \lceil d_{y_i}/\sqrt{m}\rceil$\;
}
\For{$j=1,\dots,R$}{ 
    Let $Z_j \leftarrow 0$\;
    Sample a vertex $z_j\in N(y_i)$\;
    \If{$z_j$ and $\wedge_i$ form a butterfly $\btf_j$ and $x_i \prec z_j$}
    {   
    \color{cssblue}%
    \If{\heavy$(e_i, \bar{b}, \bar{w}, \epsilon, m)$ is light}{
        Let $\ell_j\leftarrow$ number of light edges of $\btf_j$\;
        Let $Z_j\leftarrow \max(\sqrt{m}, d_{y_i})\cdot \frac{1}{\ell_j}$\;
        }
    }
    
}
$Y_i \leftarrow \frac{1}{R}\sum_{j=1}^R Z_j$\;
}
$X\leftarrow \frac{m}{s_1 s_2}\cdot\left(\sum_{(u,v)\in S}d_{(u,v)}\right)\cdot\left(\sum_{i=1}^{s_2}Y_i\right)$\;
\Return{X}\;
\end{algorithm}
From the lemma, we can observe that by setting the size of the representative edge set $S$ to be at least $c\cdot m\log(n/\epsilon^2)/\bar{b}^{\frac{1}{4}}$ where $c$ is a sufficiently large constant, then $S$ is good with high probability.
We defer the proof to Appendix X.A of our {technical} report \cite{fullversion}.



Based on the above facts, we propose \btfe in Algorithm~\ref{alg:estimate_s1_s2}. 
The algorithm follows our general two-level sampling framework and embeds it with the heavy-light technique. 
The parts of invoking \heavy subroutine are highlighted in blue. 
The correctness and efficiency are shown in the following theorem.
We refer the reader to Appendix X.B of our {technical} report \cite{fullversion}.

%

\begin{restatable}[Theoretical guarantee for \btfe]{theorem}{ThmBtfe}
\label{thm:main_result_of_estimate_s1_s2}
If the sampled edge set $S$ is good and $\bar{b}$ and $\bar{w}$ satisfy Assumption~\ref{aspt:b_m_bound}, then the output $X$ of Algorithm~\ref{alg:estimate_s1_s2} satisfies:
\begin{itemize}
    \item $\E[X] \in [b\cdot (1-c_H\cdot \epsilon), b]$. 
    \item $(1 - 3 c_H \cdot \epsilon)\cdot b < X < (1 + 3 c_H \cdot \epsilon)\cdot b$ with high probability of at least $1-\frac{2\epsilon}{\log n}$.
\end{itemize}
Meanwhile, when $\bar{b} \ge b/2$ and $\bar{w}$ satisfies Assumption~\ref{aspt:b_m_bound}, \btfe needs a query and running time complexity of $O^*(w\sqrt{m}/\bar{b}+ m/\bar{b}^{\frac14})$ if we have a uniform edge sampler. 
\end{restatable}


\begin{proof}[Proof of efficiency]
First, if a uniform edge sampler is allowed, the time cost of sampling $s_1$ edges is given by $O(s_1) = O^*(m/\bar{b}^{\frac14})$.

Next, we consider the total time cost of iterations where \heavy is not invoked.
The analysis of time complexity is similar to the proof of \Cref{lem:running_time_of_heavy}, which takes $O^*(\sqrt{m}w/\bar{b})$.


Then we calculate the expected number of calls of \heavy.
Observe that \heavy is invoked when $z_j$ and $\wedge_i$ form a butterfly and $x_i \prec z_j$, of which the probability is at most $b(\wedge_i)/d_{y_i}$.
There are $R$ rounds and the expectation of $R$ is at most $2d_{y_i}/\sqrt{m}$.
Conditioned on selecting $\wedge_j$, the expectation of the number of calls of \heavy at iteration $i$ is $(b(\wedge_i)/d_{y_i})\cdot (2d_{y_i}/\sqrt{m}) = 2b(\wedge_i)/\sqrt{m})$.
By summing them up and getting the average, the expectation of the number of calls of \heavy in each iteration $i$ is given by 
$$
\frac{\sum_{\wedge_i} 2b(\wedge_i)/\sqrt{m}}{2w}\le \frac{8b/\sqrt{m}}{2w} = O\left(\frac{b}{\bar{w}\sqrt{m}}\right)\,.
$$
Since there are $s_2$ rounds in total, the expectation of total calls of \heavy is given by
\begin{align}\label{eqn:num_of_call_of_heavy_in_s2}
O\left(\frac{b}{w\sqrt{m}}\right)\cdot O^*\left(\frac{\sqrt{m}\bar{w}}{\bar{b}}\right) = O^*(1)\,.
\end{align}
Since each call of \heavy costs $O^*(\bar{w}\sqrt{m}/\bar{b})= O^*(w\sqrt{m}/b)$ time in expectation, by summing the time cost of two parts, the expectation of time cost of the $s_2$ iterations is given by 
$O^*(w\sqrt{m}/b)$.
Therefore, we can complete the proof of efficiency by summing the time cost of sampling $s_1$ edges and running $s_2$ rounds.
\end{proof}

When the assumption $\bar{b} \ge b/2$ is relaxed to $\bar{b}\ge b/T$ for some constant $T$, the expectation of calls of \heavy (\Cref{eqn:num_of_call_of_heavy_in_s2}) changes to $O^*(T)$ correspondingly and we immediately get the following corollary:
\begin{corollary}\label{coro:bound_of_effi_when_bbar_ge_bdivT}
Given a constant $T> 0 $, when $\bar{b} \ge b/T$, the query and time complexity of \btfe is at most $T$ times the upper bound in \Cref{thm:main_result_of_estimate_s1_s2}.
\end{corollary}

\subsection{Guess-and-Prove}
\label{sec:guess_and_proof}
\begin{algorithm}
\caption{Finalized Algorithm with Guess-and-Prove (\btff)}
\label{alg:guess_and_proof}
\SetKwProg{Fn}{Function}{:}{}
\Fn{\btff$(\epsilon, n, m, G)$}{
Let $\epsilon\leftarrow \epsilon/3c_H$ and $\tilde{b}\leftarrow n^4$\;
Apply Feige's algorithm~\cite{feige2004sums} to get an estimate $\bar{w}$\label{line:approx_w}\;
\While{$\tilde{b}\ge 1$}{ \label{line:outer_loop_of_tilde_b}
    \For{$\bar{b}=n^4,n^4/2,n^4/4$ to $\tilde{b}$}{ 
            \codecmt{Proof phase starts} \\ 
            \label{line:proof_phase_starts}
            \For{$i=1$ to $c\epsilon^{-1}\log(\log n)$ } 
            {     
                    Let $X_i\leftarrow$ \btfe$(\bar{b}, \bar{w}, \epsilon, m)$\;
                }
                Let $X\leftarrow \min_i X_i$\;
                \If{$X\ge \bar{b}$}{
                    \codecmt{Proof phase ends}\\
                    \Return{$X$};
                    \label{line:proof_phase_ends}
                }
            }
            Let $\tilde{b}\leftarrow \tilde{b}/2$\;
        }
}
\end{algorithm}
Here we present our final algorithm that involves the heavy-light partition and guess-and-prove techniques, as shown in Algorithm~\ref{alg:guess_and_proof}.
Firstly, we initialize $\tilde{b}$ as $n^4$, which is used to control the time complexity.
It invokes Feige's algorithm~\cite{feige2004sums} to get an estimate $\bar{w}$ satisfying Assumption~\ref{aspt:b_m_bound}.
Then it searches for an ideal $\bar{b}$ through a nested loop.
For each $\bar{b}$, we run the ``prove'' phase from Line~\ref{line:proof_phase_starts} to Line~\ref{line:proof_phase_ends}.
It runs the estimate algorithm ${\tt TLS\text{-}EG}(\bar{b}, \bar{m}, \epsilon, G)$ several times and takes $X$ as the minimum among the estimated values.
When $\bar{b}$ satisfies Assumption~\ref{aspt:b_m_bound}, then $X \ge \bar{b}$ holds with high probability, and the ``prove'' phase will terminate.

Notably, when $\bar{b}$ reaches a value less than $b/2$, by Corollary~\ref{coro:bound_of_effi_when_bbar_ge_bdivT}, the expected running time of \btfe cannot be bounded as tightly as \Cref{thm:main_result_of_estimate_s1_s2}.
To address this issue, we introduce the outer loop at Line~\ref{line:outer_loop_of_tilde_b}, where $\bar{b}$ gradually decreases from $n^4$ to $\tilde{b}$ within each round and $\tilde{b}$ is then halved at the end of each outer round.
The theoretical guarantee provided by \btfe (as stated in Lemma~\ref{thm:main_result_of_estimate_s1_s2}) shows that the probability of $\tilde{b}$ being reduced to a very low value is exceedingly small. Hence, we will not run \btfe with $\bar{b}$ less than $b/2$ with high probability.




\begin{theorem}[Correctness of the finalized algorithm]\label{lem:correctness_of_guess_and_proof}
Algorithm~\ref{alg:guess_and_proof} returns a value $X$ satisfying $(1-\epsilon)b\le X\le (1+\epsilon)b$ with a probability of $5/6$.
\end{theorem}
\begin{proof}
We say $X$ is \emph{correct} if $X$ satisfies $(1-\epsilon)b\le X\le (1+\epsilon)b$.
Now we first consider a simple case when $\tilde{b}$ is fixed as a constant $1$, rather than decreasing in the outer loop.
Hence, $\bar{b}$ will traverse from $n^4$ to $1$. 
Then we claim the following two facts.
\begin{enumerate}
    \item When $\bar{b} > 2b$, the probability that the inner loop continues is at least $1-1/\log^3n$.     
    \item When $\bar{b}$ satisfies the first condition of Assumption~\ref{aspt:b_m_bound}, Algorithm~\ref{alg:guess_and_proof} will halt and return a correct estimate with a probability of $1-(2\epsilon/3c_H)/\log n$.
\end{enumerate}
The proofs of these two facts are omitted due to the page limit.
%
Now we consider the general case when $\tilde{b}$ is not fixed as the constant $1$.
When $\tilde{b} > 2b$, according to item (1), the probability that Algorithm~\ref{alg:guess_and_proof} throws an incorrect answer is at most 
\begin{gather}\label{eqn:prob_return_wrong_ans}
\frac{1}{\log^3n}\cdot \sum_{k=0}^{\log n} \left(1-\frac{1}{\log^3n}\right)^k  \le \frac{1}{\log^2n}\,. 
\end{gather}
\enlargethispage{\baselineskip}
For $\tilde{b} \le 2b$, the probability that $\bar{b}$ reaches a value satisfying $b/2 <\bar{b} < 2b$ is at least 
$$
\left(1- \frac1{\log^3n}\right)^{\log n} \ge 1- \frac{1}{\log^2n} \ge \frac89 \,, \text{ when $n>20$\,.}
$$
Then by item (2), \btfe returns a correct answer when $\bar{b}$ reaches $b/2 <\bar{b} < 2b$ is at least $1-(2\epsilon/3c_H)/\log n$. 
Combining them with Eqn~(\ref{eqn:prob_return_wrong_ans}), the probability that Algorithm~\ref{alg:guess_and_proof} returns a correct answer is at least 
$$
\left(1- \frac{1}{\log^2n}\right)^{\log n} \cdot \frac89 \cdot \left(1-\frac{(2\epsilon/3c_H)}{\log n} \right) \ge \frac56,
$$
when $n > 20$.
\end{proof}

\begin{theorem}[Efficiency of The Finalized Algorithm]
\label{thm:correctness_of_guess_and_proof}
If we allow an edge sampler, Algorithm~\ref{alg:guess_and_proof} needs a query and time complexity of $$O^*(w\sqrt{m}/{b} + m/b^{\frac14})\,.$$ 
\end{theorem}
\begin{proof}
    The processing of Feige's algorithm only takes $O(m/\sqrt{w})$. Since $b \leq O(w^2)$, it takes $O(m/b^{\frac14})$ time and query complexity. 
    By \cref{thm:main_result_of_estimate_s1_s2}, the query and time complexity for each iteration is $O^*(w\sqrt{m}/{b} + m/b^{\frac14})$. The outer loop takes at most $O(\log n)$ times, and the inner loop takes at most $O(\epsilon^{-1} \log n)$ times. In conclusion, the total time and query complexity are still $O^*(w\sqrt{m}/{b} + m/b^{\frac14})\,.$
\end{proof}
\section{Experiments}
\label{sec:experiments}
\enlargethispage{1\baselineskip}  
In this section, we evaluate the performance of our method and the existing algorithms.
Our experiments aim to answer the following questions: 
\begin{itemize}[leftmargin=*]
    \item \textbf{RQ1 (Running cost):} How does \our reduce the \emph{query cost} and  \emph{running time} compared to the existing solutions? 
    \item \textbf{RQ2 (Accuracy):} Does \our still provide \emph{robust} and \emph{accurate} estimates?
    \item \textbf{RQ3 (Ablation study):} How do the parameters of \our affect its performance?
\end{itemize}
We answer the first two questions, RQ1 and RQ2, by running the three methods on {15} real-world datasets, where our method \our substantially reduces the running cost while ensuring comparable relative errors compared to existing solutions {(with an average reduction of more than $85.68\%$ on the number of queries and $65.14\%$ on the running time).}
Finally, we answer the last question, RQ3, by choosing several representative datasets and varying the input parameters, leading to a practically good setting for the parameters.


\subsection{Experimental settings}
\begin{table}[t]
\small
\setlength{\tabcolsep}{2pt}
\caption{Summary of datasets}
\vspace{-2mm}
\label{table:datasets}
\begin{center}
\begin{tabular}{|l||c|c|c|c|}
\hline
\textbf{Dataset} & $|E|$  & {\bf {$|U|$}} & $b$ & {Density}  \\
\hhline{=====}
Amazon  & {5,743,258} & {2,146,057} & {1,230,915}  &3.54 \\
\hline
Wiki-He & 6,902,392 & {892,518} & 98,906   & 23.23 \\
\hline
Movielens  & {10,000,054} & {69,878} & {10,677}  & 366.00  \\
\hline
DBLP  & {12,282,059} & {1,953,085} & {5,624,219}  & 0.004 \\
\hline
Wiki-Ja & 21,418,738 & {3,126,385} & 444,563   & 17.88 \\
\hline
Wiki-Fr &  52,950,008 & {8,870,762} & 757,621& 20.42 \\
\hline
Wiki-De &  55,231,903 & {1,025,084} & 5,910,432 &  22.44\\
\hline
Reuters(Reu) &  60,569,726 & {781,265} & 283,911 & 128.61\\
\hline
Bag-Ny & 69,679,427 &{101,636} & 299,752  & 399.20 \\
\hline 
Delicious(Deli) & 81,989,133 & {833,081} & 4,512,099  &  42.35  \\
\hline 
Gottron(Gott) & 83,629,405 & {556,077} & 1,173,225  &  103.54 \\
\hline 
Bi-LiveJournal(Bi-Lj) & 112,307,385 & {3,201,203} & 7{,}489{,}073   & 22.94 \\
\hline
Yahoo & 256{,}804{,}235 & {624,961} & 1,000,990 & 324.68 \\
\hline
Orkut  & 327{,}037{,}487  & {2,783,196} & 8{,}730{,}857  &  66.34\\
\hline
Wiki-En  & {572,591,272} & {42,640,545} & {8,116,897}  &  6.56\\
\hline

\end{tabular}
\end{center}
\vspace{-6mm}
\end{table}

\noindent{\bf Datasets.}  
{ 
We collect 15 real-world network datasets from KONECT~\cite{konect}, with edge counts from $6 \times 10^6$ to $6 \times 10^8$ and butterfly counts from $3 \times 10^{7}$ to $1 \times 10^{14}$}, covering all datasets used in the previous work~\cite{zhang2023scalable}. 
Dataset statistics are in \Cref{table:datasets}, where $|E|$ is the edge count, $L$ and $U$ are vertex layers of the bipartite graph, ``Density'' (defined by $m/\sqrt{\abs{L}\cdot \abs{U}}$) reflects the density of the graph.


\medskip\noindent{\bf Algorithms.} We compare our algorithm \our (Algorithm~\ref{alg:toy_our_main_algo}) with two state-of-the-art methods, \naive and \bs (\Cref{sec:sota}), implemented as follows: 
\textbf{(a)} \naive takes a probability $p$, samples a subgraph accordingly, applies exact counting, and scales the result to estimate the original graph.
\textbf{(b)} \bs samples $r$ vertex pairs, estimates butterflies induced by each pair, and returns their average.


\medskip\noindent{\bf Query model.} 
{ 
Our evaluation considers the same types of queries as our theoretical model, including degree query, neighbor query, and vertex-pair query, also an edge uniform sampler with one query cost each time. 
In practice, storing past query results can consume space proportional to the graph size, while retrieving and utilizing them also introduce extra computation.
Thus, we assume that the algorithms do not store the results of past queries.
}

\medskip
\noindent{\bf Evaluation metrics.} We measure running cost by query count and running time. The relative error of an estimate $\bar{b}$ is defined as $(\bar{b} - b)/b$, where $b$ is the true value.
{We use box plots to illustrate algorithm accuracy and robustness, with boxes for 90\% confidence intervals, whiskers for min/max values, and dashed lines for medians. For clarity of presentation, values exceeding the $y$-axis are truncated and marked with ``//''.
Each algorithm runs 50 times per dataset. We report the mean running time, query count, and box plots of relative errors. Following \cite{zhang2023scalable}, it might not be appropriate to report the mean or median of the estimates $\bar{b}$ across the 50 runs. It would artificially simulate an algorithm with 50× cost and reduced variance. Instead, we analyze estimate distributions, showing 90\% confidence intervals and min/max relative errors.}

\medskip
\noindent{\bf Parameter settings.} 
a) For \naive, we adopt the same setting of $p=0.2$ as~\cite{sanei2018butterfly};
b) For \bs, we run the pair-sampling method for $r= 2\cdot 10^4$ rounds.
Despite $r$ being fixed, we compare \bs and \our on the basis that they achieve the same level of relative errors.
It should be noted that we do not deliberately set $r$ as a large number.
As demonstrated by \cite[Figure 5]{sanei2018butterfly}, the relative errors of \bs decrease nearly linearly in the initial ten seconds of \bs, and our implementation of \bs terminates within ten seconds on almost all datasets, as shown in \Cref{fig:runtime}.
Hence, a smaller $r$ will cause a proportional increase in the relative errors of \bs; 
{Moreover, under fixed time or query budgets, our algorithm also demonstrates superior accuracy compared to \bs, which will be discussed in the following paragraph.
}
{
c) For our method \our, three parameters are involved: the size of the representative edge subset $s_1$ per iteration, inner sampling rounds $s_2$, and outer sampling rounds $r$. We set $s_1 = 0.5\sqrt{m}$ for experiments, where $m$ is the edge count of graph $G$. The impact of varying $s_1$ is shown later. Parameters $s_2$ and $r$ are set via an automatic termination strategy:
\enlargethispage{\baselineskip}
\begin{itemize}[leftmargin=*]
    \item Inner rounds ($s_2$): \our performs the sampling in batches of size $0.1 \times \sqrt{m}$ and terminates it when the change of estimate becomes small. 
    Particularly, the inner loop stops when the latest batch alters the estimate by less than $2\%$.
    \item Outer rounds ($r$): Similarly, the process terminates if the current round changes the estimate by less than $0.2\%$.
\end{itemize}
}

\medskip\noindent{\bf Environment.} All our algorithms are implemented in C++ and executed on an Ubuntu machine equipped with an Intel(R) Xeon(R) Gold 6342 CPU (96 cores @ 2.80GHz) and 512GB main memory. For all experiments, we run the algorithms in a single-threaded setting without parallelization.

\begin{figure*}[t]
    \vspace{-6mm}
    \centering
    \begin{subfigure}[t]{0.32\textwidth}
        \centering
        \includegraphics[width=\linewidth]{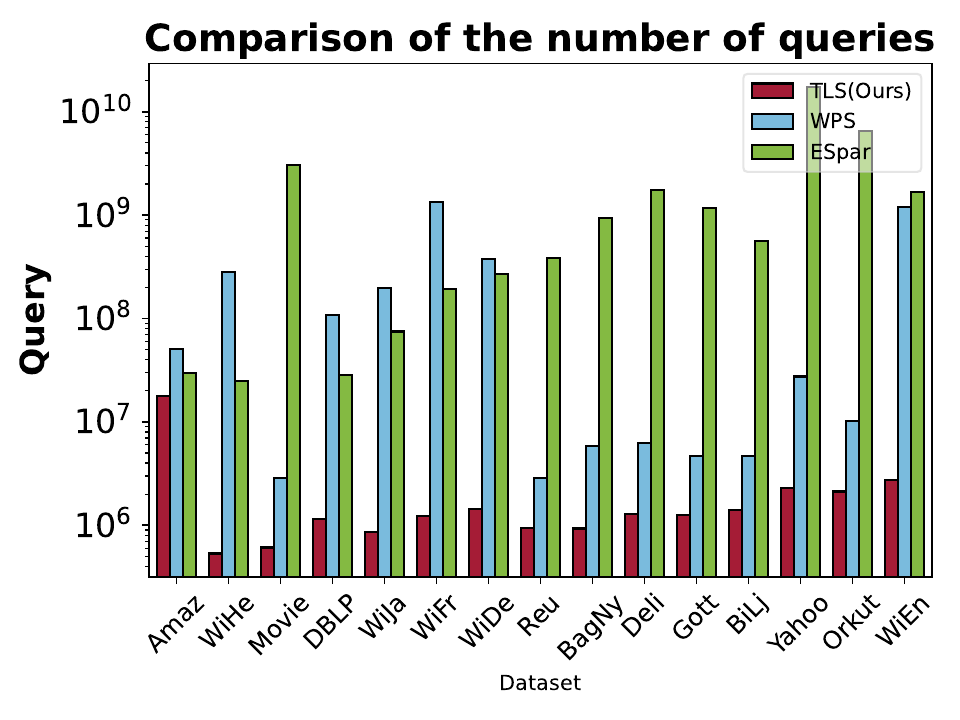}
        \vspace{-7mm}
        \caption{Comparison of the number of queries.}
        \label{fig:num_of_query}
    \end{subfigure}
    \hfill
    \begin{subfigure}[t]{0.32\textwidth}
        \centering
        \includegraphics[width=\linewidth]{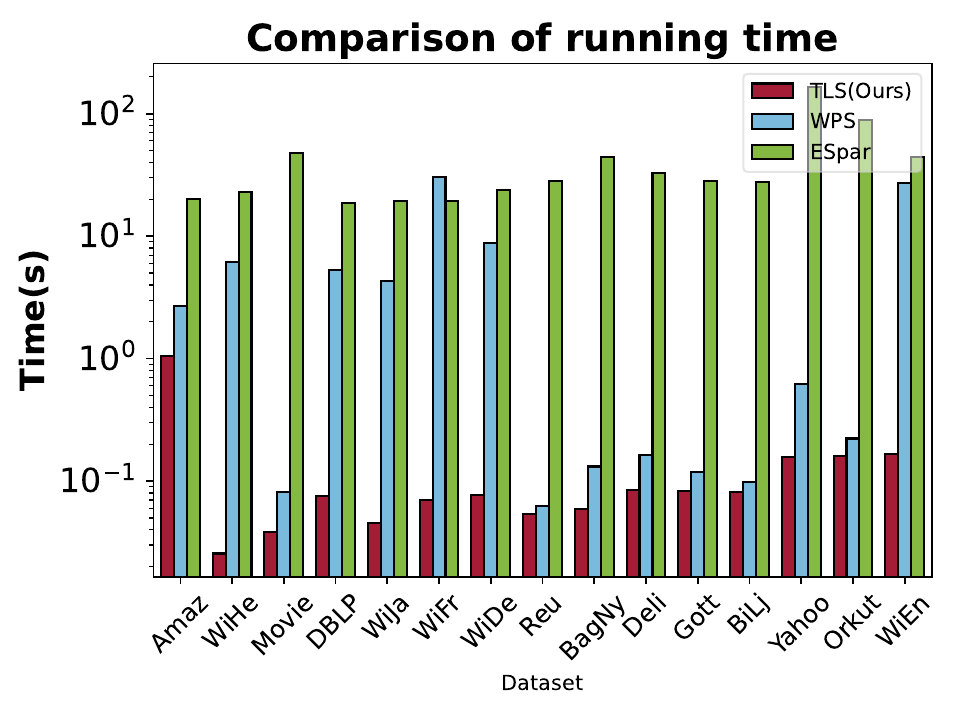}
        \vspace{-7mm}
        \caption{Comparison of running time.}
        \label{fig:runtime}
    \end{subfigure}
    \hfill
    \begin{subfigure}[t]{0.32\textwidth}
        \centering
        \includegraphics[width=\linewidth]{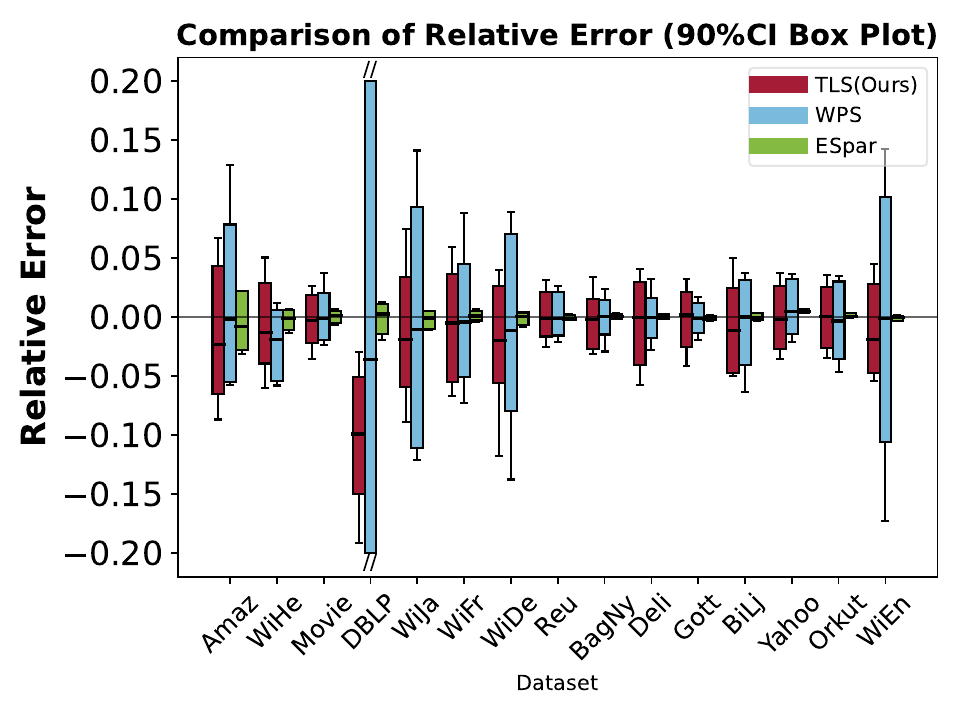}
        \vspace{-7mm}
        \caption{Comparison of relative error. 
        }
        \label{fig:main_rela_error}
    \end{subfigure}
    \caption{Overall comparison of different metrics.}
    \label{fig:overall_comparison}
    \vspace{-5mm}
\end{figure*}

\begin{figure*}[htbp]
\centering
\begin{minipage}[t]{1\textwidth}
\centering
\includegraphics[width=4.3cm]{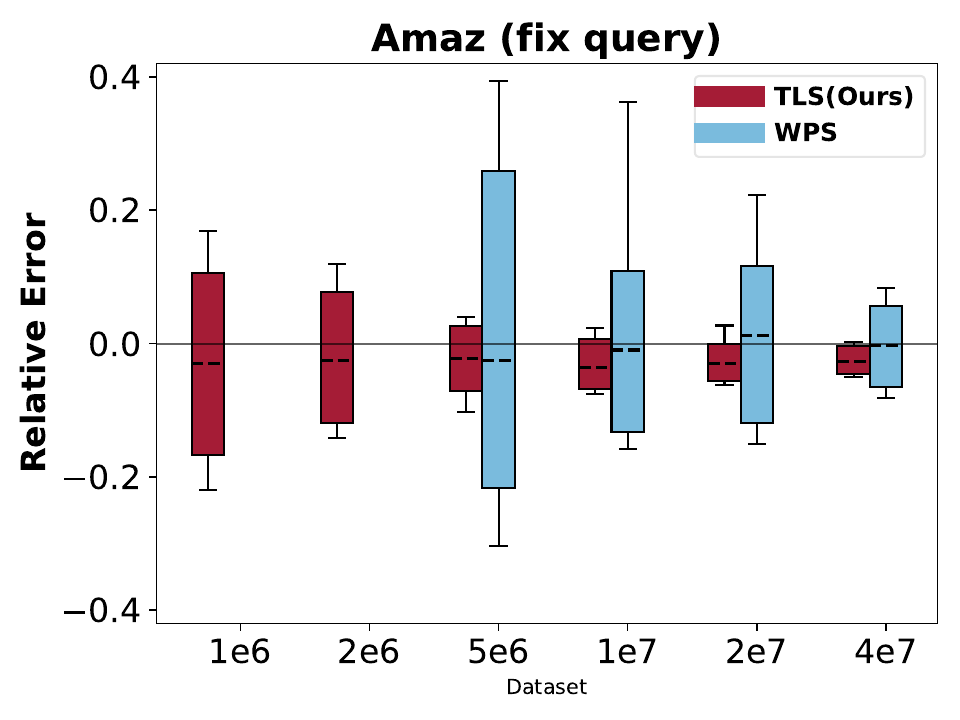}
\includegraphics[width=4.3cm]{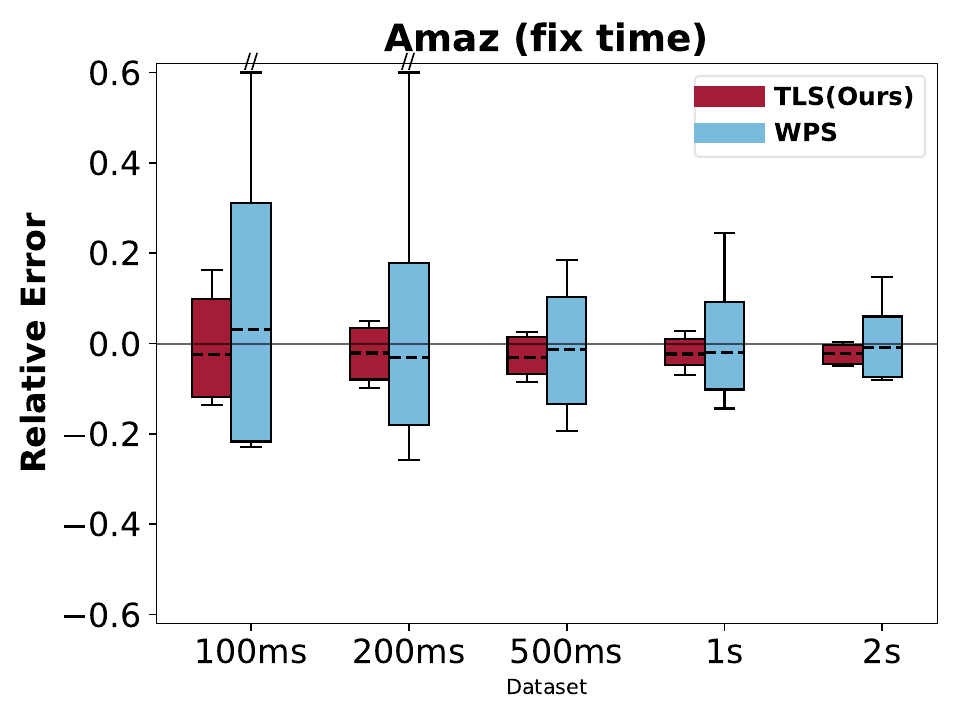}
\includegraphics[width=4.3cm]{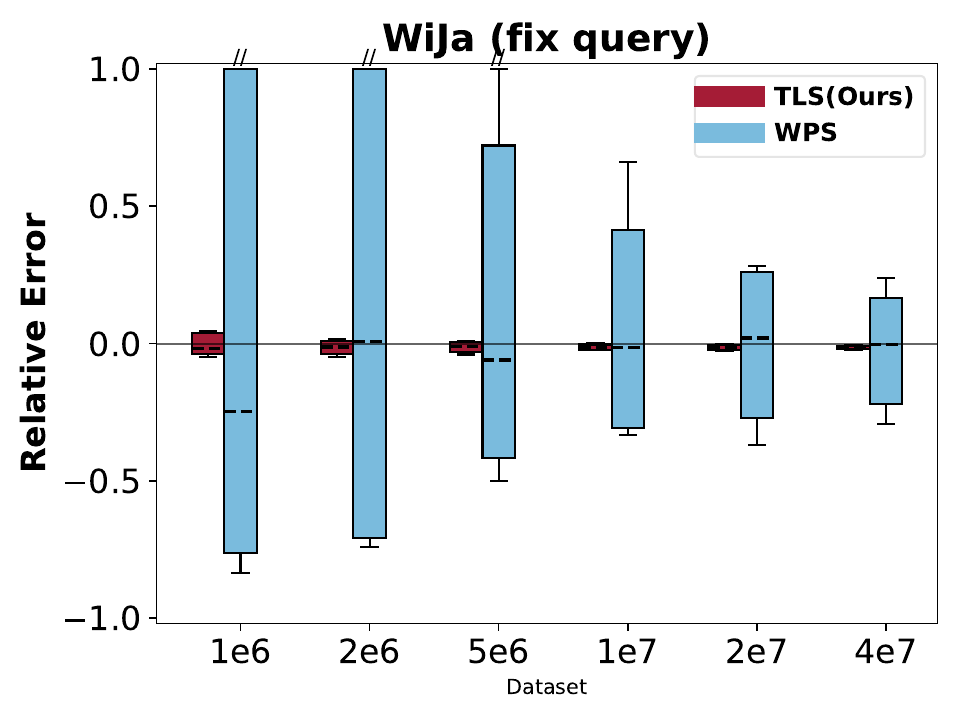}
\includegraphics[width=4.3cm]{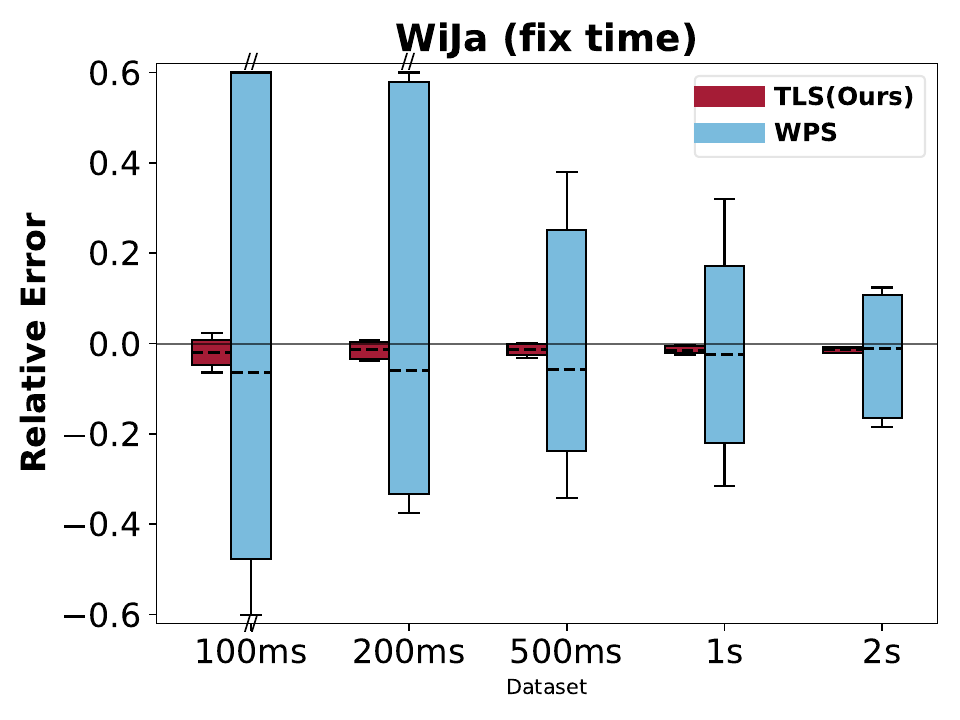}
\includegraphics[width=4.3cm]{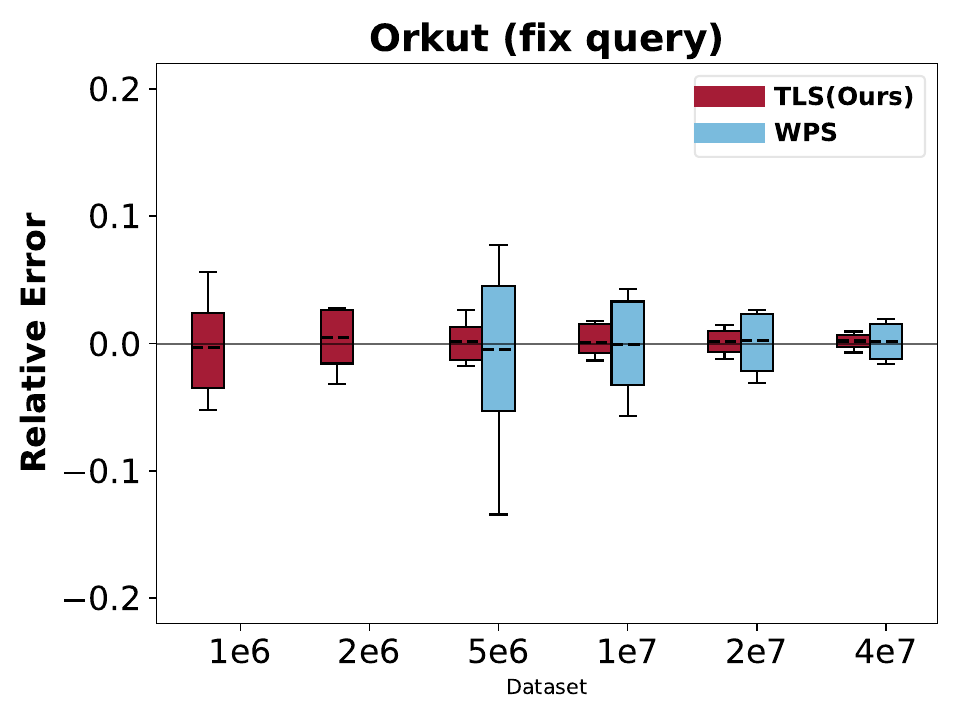}
\includegraphics[width=4.3cm]{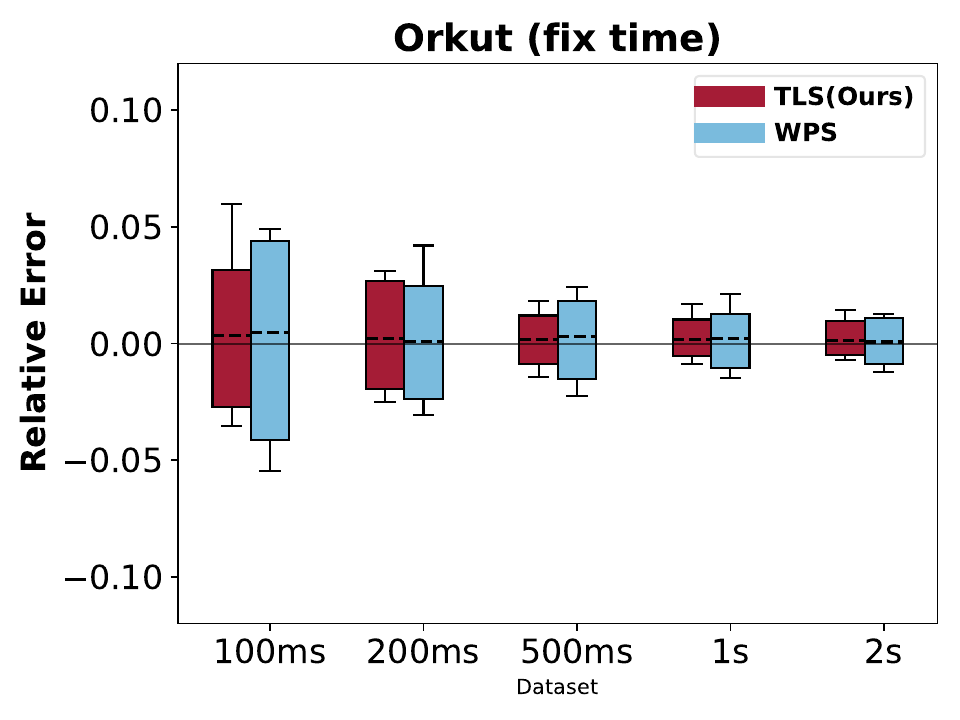}
\includegraphics[width=4.3cm]{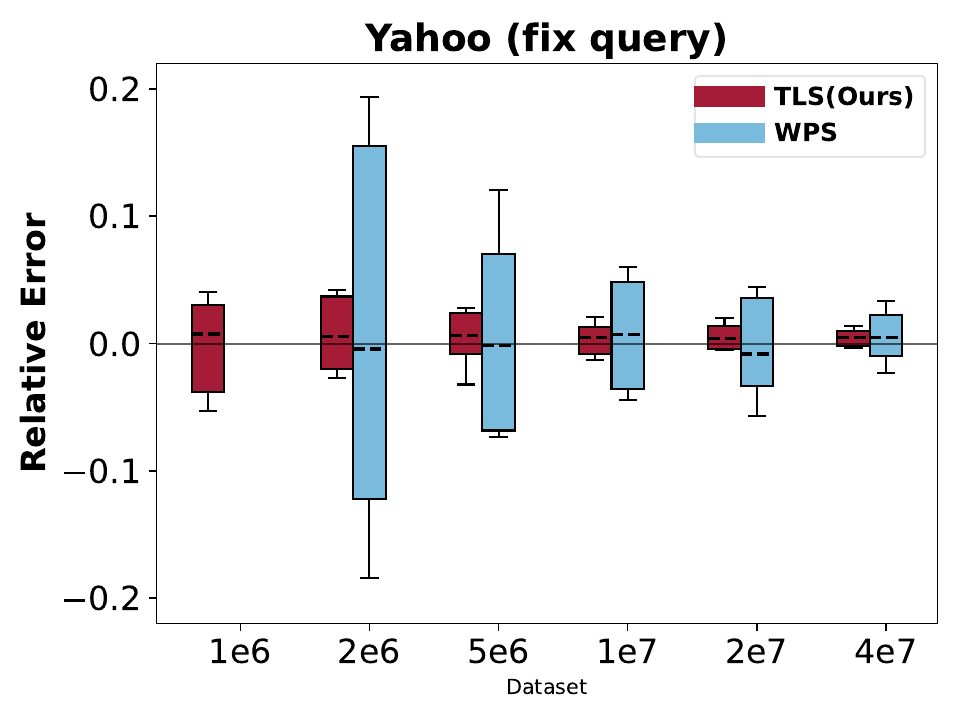}
\includegraphics[width=4.3cm]{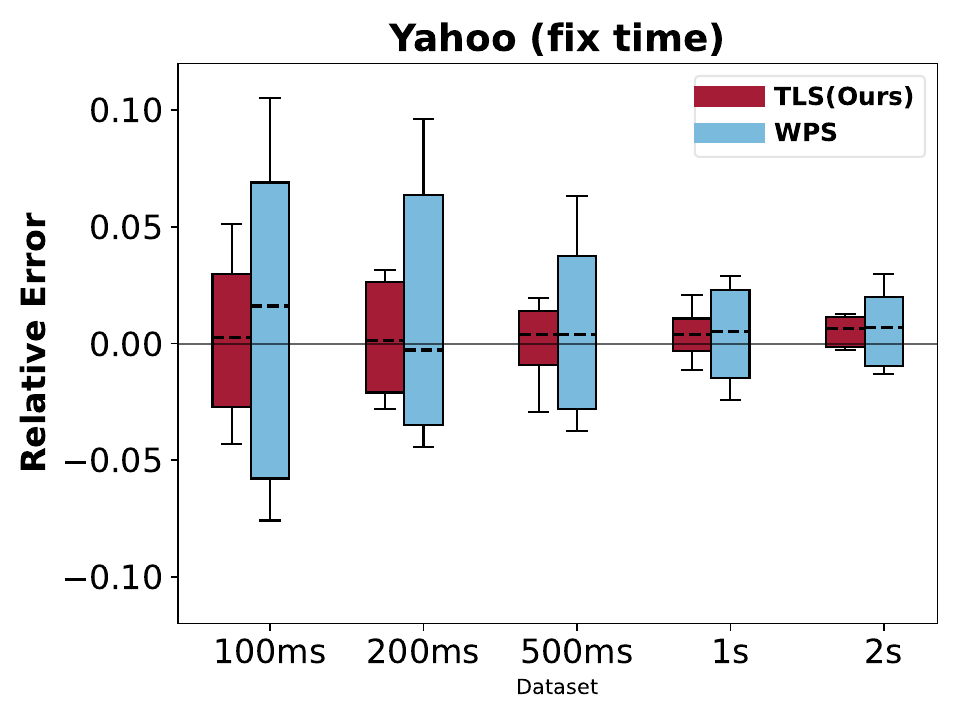}
\end{minipage}

\begin{minipage}[t]{1\textwidth}
\centering
\caption{Relative errors under fixed time/query.}
\label{fig:fix_time_query}
\end{minipage}
\vspace{-7mm}
\end{figure*}



\subsection{Experimental results} \label{sec:res}
\noindent{\bf Overall performance.} 
\Cref{fig:num_of_query} and \Cref{fig:runtime} plot the mean values of the number of queries and running time, while \Cref{fig:main_rela_error} {demonstrates box plots with the $90\%$ confidence intervals of relative errors.}
From the figures, we can obtain the following observations.

\our significantly reduces the running cost and the number of queries.
In \Cref{fig:num_of_query} and \Cref{fig:runtime}, \our outperforms the two baseline methods regarding both queries and running time on all {15} datasets.
{ 
Compared to \bs, \our reduces queries by 85.68\% and running time by 65.14\% on average. Against \naive, \our achieves greater improvements, cutting queries by 95.37\% and running time by 99.43\% on average.
}

Notably, \our significantly reduces costs by orders of magnitude on datasets like \texttt{Wiki-He}, \texttt{Wiki-Fr}, \texttt{Wiki-De}, \texttt{Yahoo}, and \texttt{Wiki-En}. For example, on \texttt{Wiki-He}, \our reduces queries by 523× and running time by 242× compared to \bs, achieving a significant reduction of 99.8\% and 99.6\%, respectively. On \texttt{Wiki-En}, \our achieves 435× and 164× speedups over \bs. As expected, \naive has the highest number of queries and longest running time on most datasets.


\our delivers reasonable relative errors across all datasets. As shown in \Cref{fig:main_rela_error}, absolute relative errors of confidence intervals are below 5\% on most datasets. On the extremly sparse graph \texttt{DBLP}, both \our and \bs shows a decline in accuracy, while \our still significantly outperform \bs.
\our matches \bs in accuracy; for example, on \texttt{Orkut}, \our bounds relative error within 2.5\% at 90\% confidence, versus 3.5\% for \bs. On \texttt{Yahoo}, \our achieves 2.7\% error compared to \bs's 3.2\%. Across all datasets, \our's relative error is at most 2.3\% higher than \bs's. \naive's relative errors stay below 1\% on most datasets. It can be expected that \naive has the lowest relative errors since it almost enumerates all the butterflies on the original graph and incurs large query and time costs.

\begin{figure*}[htbp]
\centering
\centering
\includegraphics[width=4.3cm]{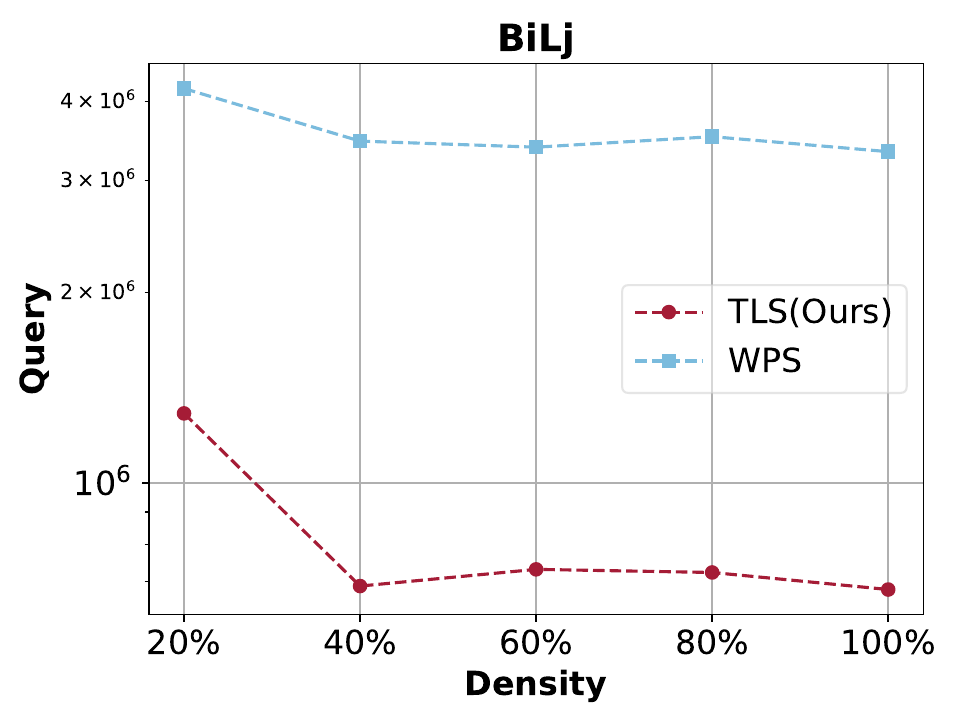}
\includegraphics[width=4.3cm]{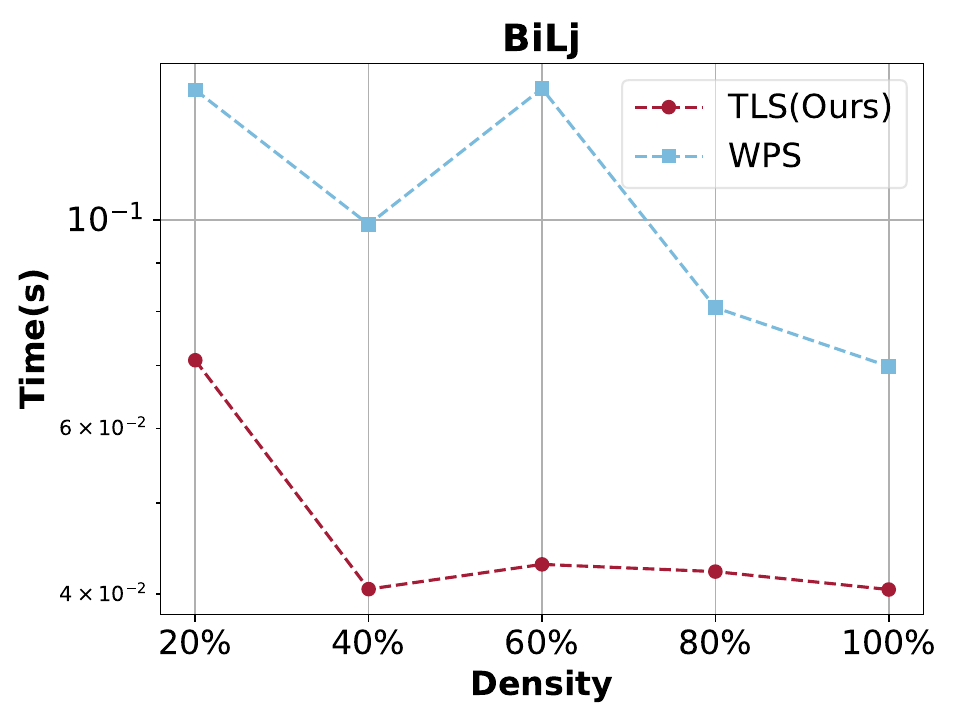}
\includegraphics[width=4.3cm]{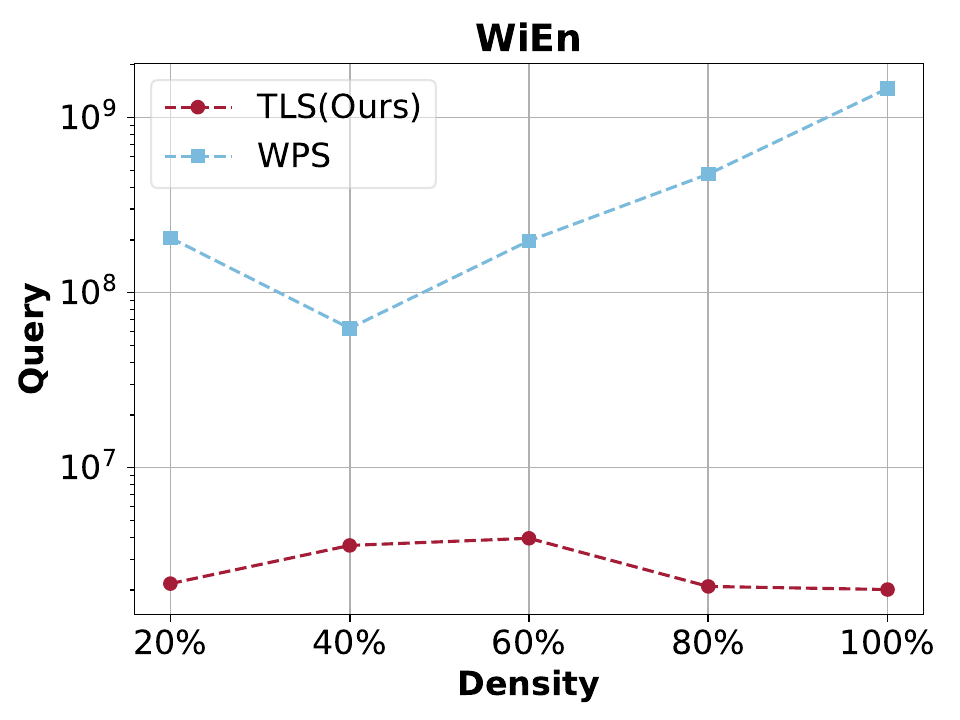}
\includegraphics[width=4.3cm]{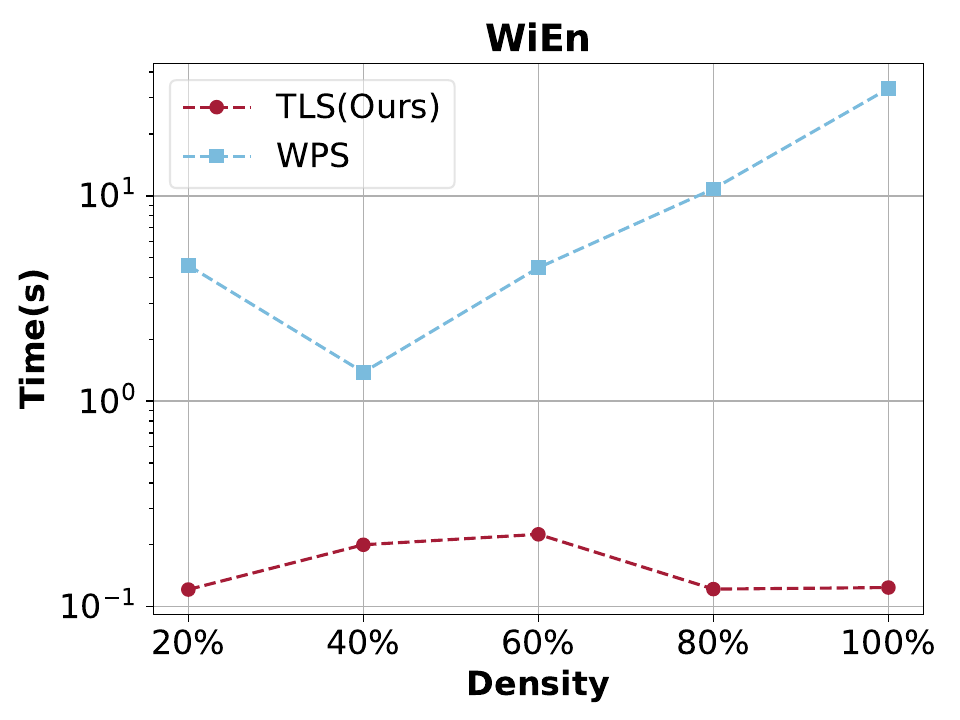}
\includegraphics[width=4.3cm]{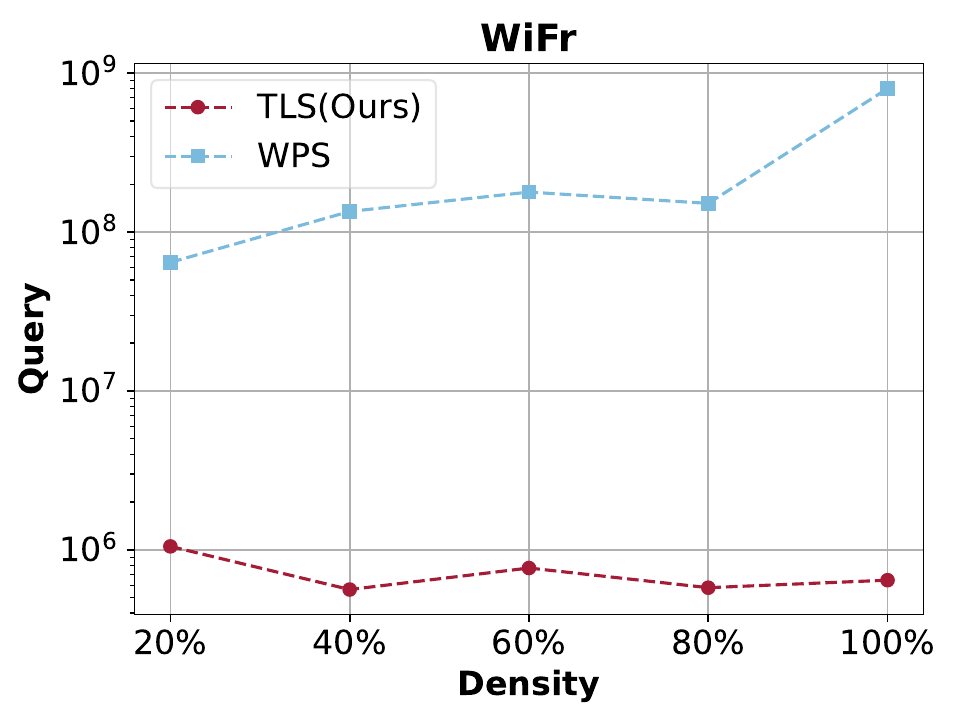}
\includegraphics[width=4.3cm]{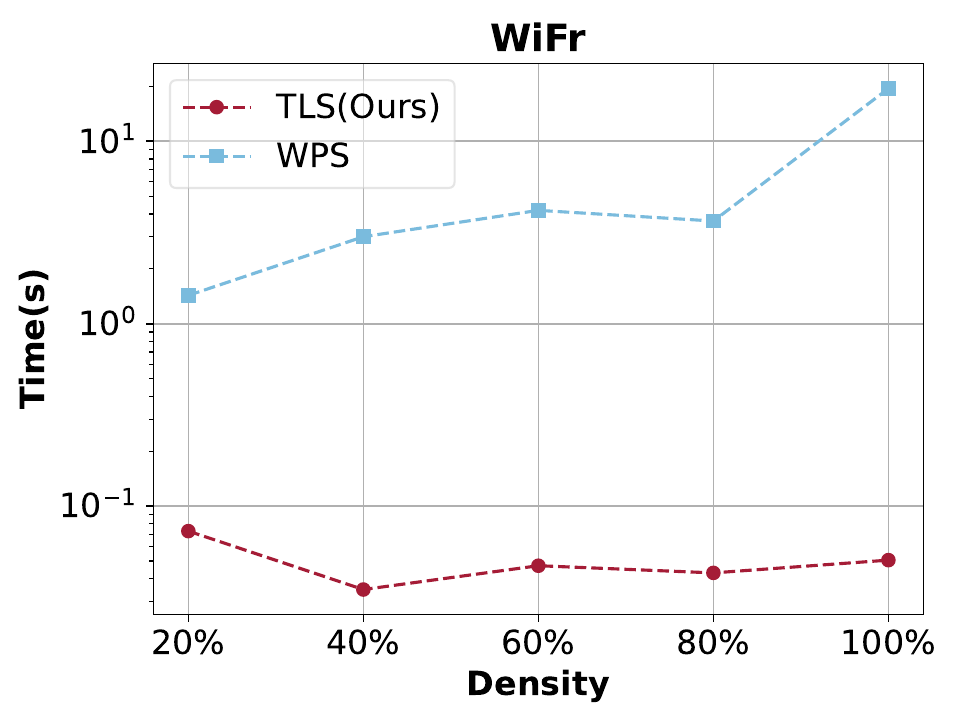}
\includegraphics[width=4.3cm]{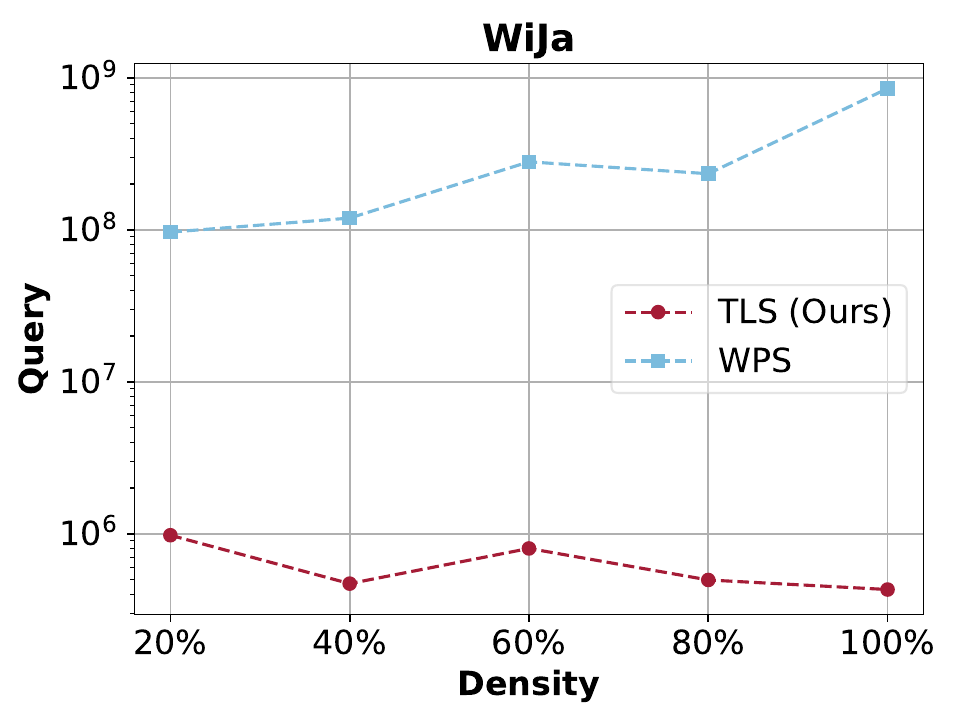}
\includegraphics[width=4.3cm]{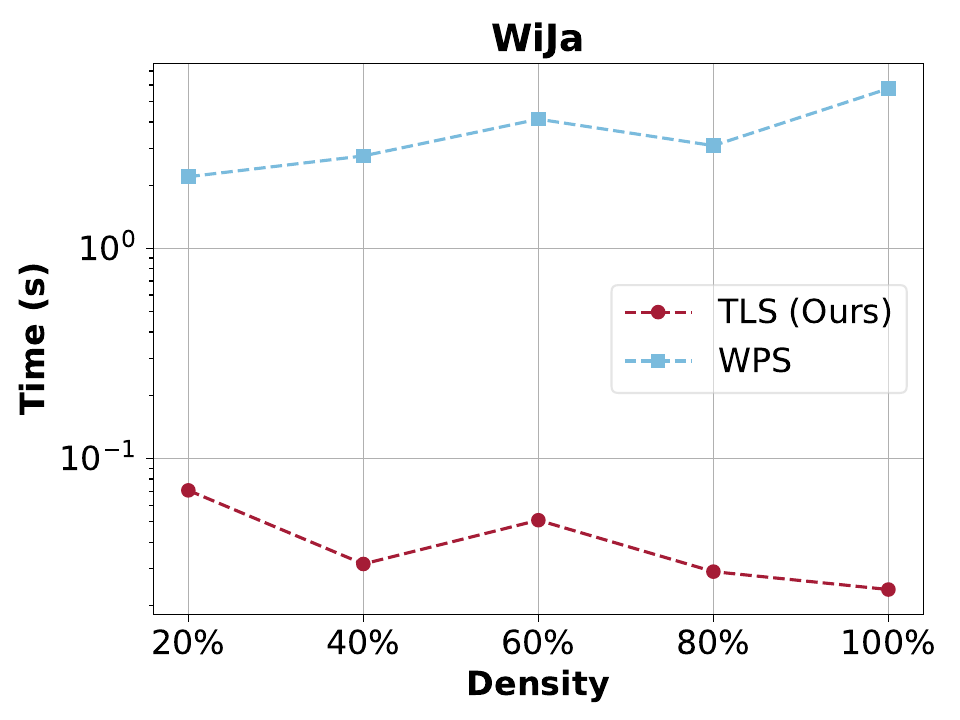}
\vspace{-2mm}
\caption{Time and query cost of obtaining 3\% relative error on varying graph density}
\label{fig:cost_of_vary_density}
\vspace{-7mm}
\end{figure*}

\begin{figure*}[htbp]
\vspace{-4mm}
\centering
\begin{minipage}[t]{1\textwidth}
\centering
\includegraphics[width=4.3cm]{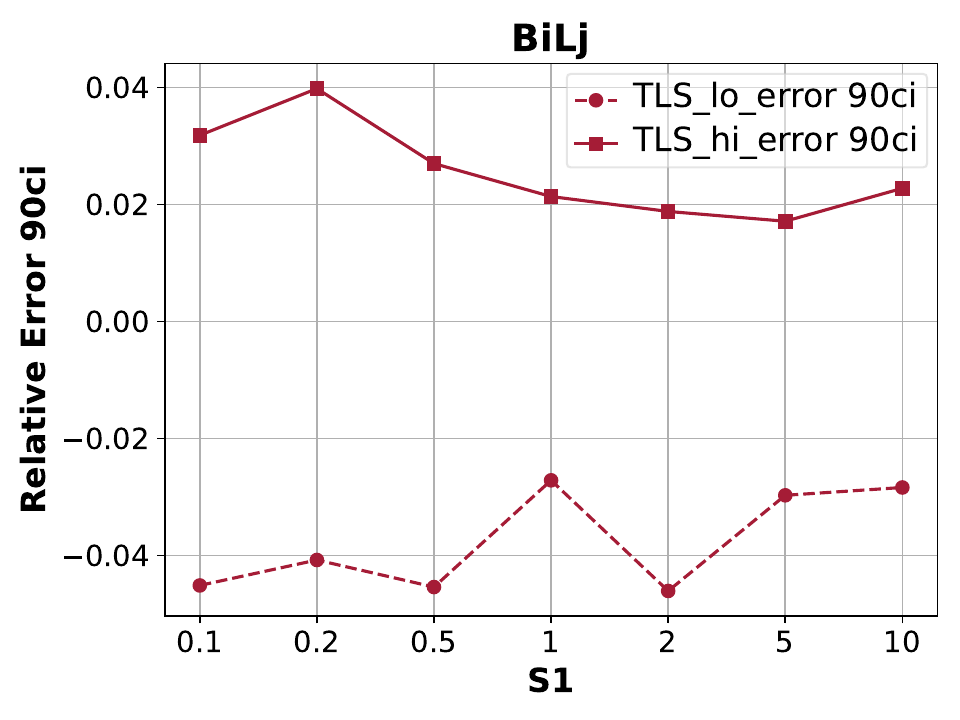}
\includegraphics[width=4.3cm]{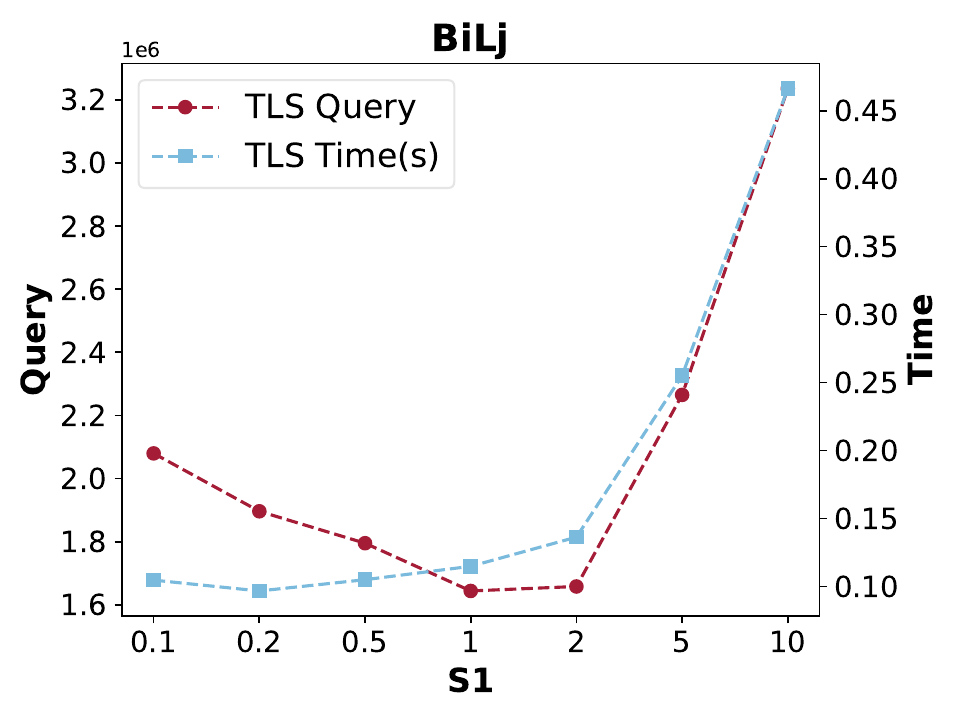}
\includegraphics[width=4.3cm]{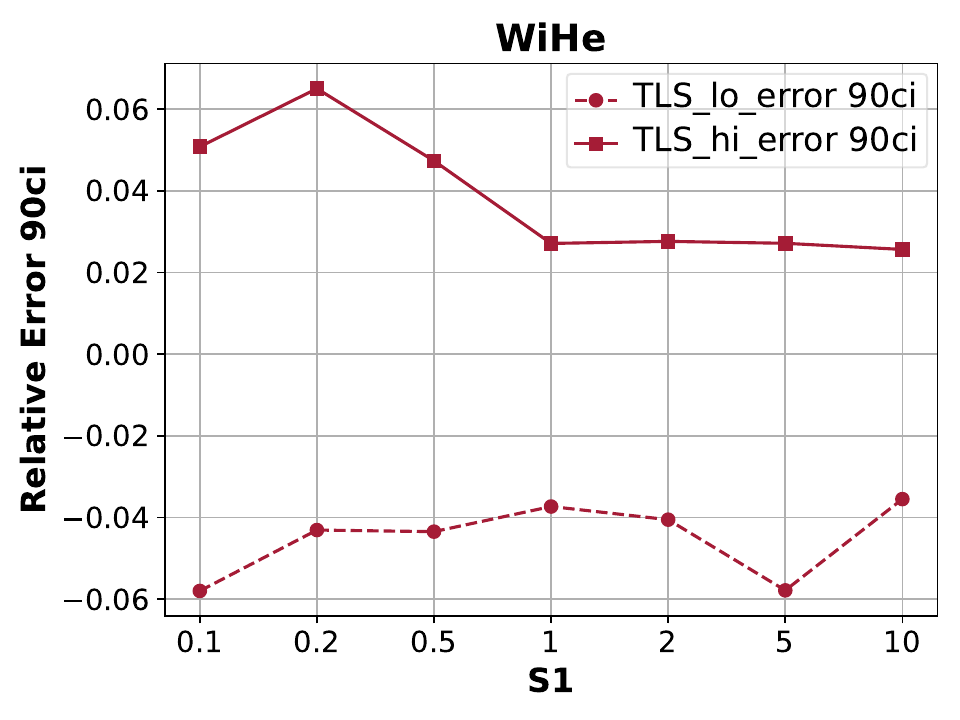}
\includegraphics[width=4.3cm]{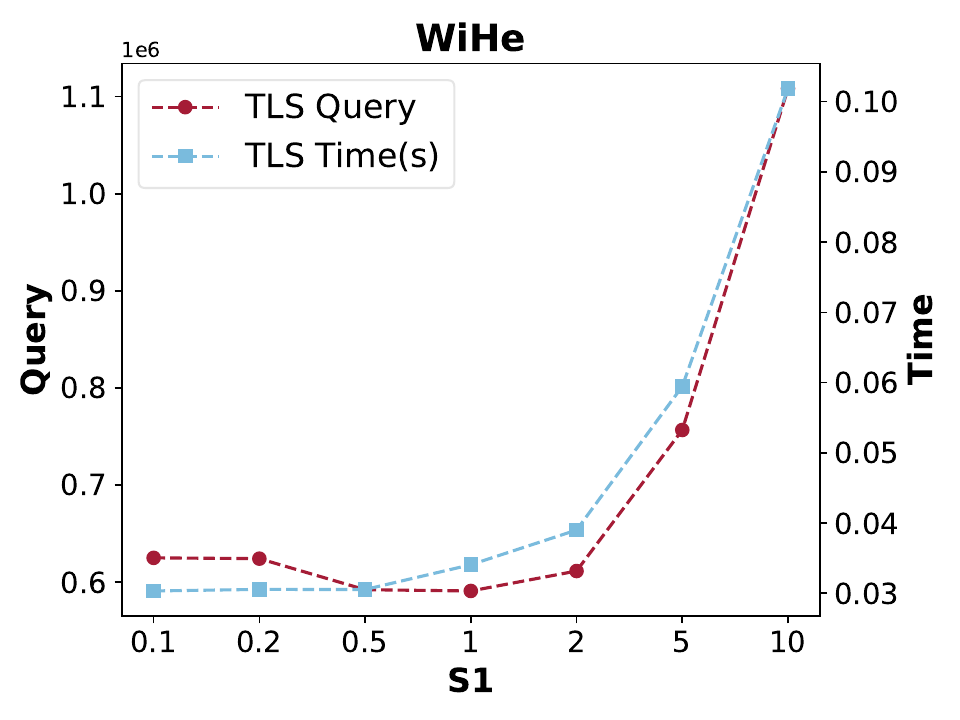}
\end{minipage}

\begin{minipage}[t]{1\textwidth}
\centering
\vspace{-5mm}
\caption{Relative errors, time, and query on varying $s_1$.}
\label{fig:cost_of_vary_s1}
\end{minipage}
\vspace{-3mm}
\end{figure*}

\begin{table*}[t]
\vspace{-3mm}
\small
\caption{{Peak memory usage (MB)}}
\vspace{-2mm}
\label{table:peak_memory}
\begin{center}
\resizebox{\textwidth}{!}{
\begin{tabular}{|l||c|c|c|c|c|c|c|c|c|c|c|c|c|c|c|}
\hline
\textbf{Method} & Amazon & Wiki-He & Movielens & DBLP & Wiki-Ja & Wiki-Fr & Wiki-De & Reuters & Bag-Ny & Delicious & Gottron & Bi-LJ & Yahoo & Orkut & Wiki-En \\
\hhline{================}
\naive & 34.54 & 18.10 & 15.87 & 76.55 & 59.92 & 154.26 & 137.17 & 100.57 & 109.38 & 165.83 & 140.81 & 252.95 & 404.29 & 586.81 & 777.39 \\
\hline
\bs    & 48.00 & 1.50  & 1.50  & 24.00 & 6.00  & 12.00  & 12.00  & 12.00  & 6.00   & 12.00  & 12.00  & 48.00  & 12.00  & 48.00  & 112.00 \\
\hline
\our(Ours) & 0.08  & 0.08  & 0.08  & 0.08  & 0.16  & 0.16  & 0.16  & 0.16  & 0.31  & 0.31  & 0.31  & 0.31  & 0.31  & 0.63  & 0.63 \\
\hline
\end{tabular}
}
\end{center}
\vspace{-7mm}
\end{table*}

\medskip

\noindent{\bf Relative error under fixed computational cost with \bs.}
{To provide a more comprehensive comparison between our algorithm \our and the existing method \bs, we conducted experiments on four representative datasets \texttt{Amazon}, 
\texttt{Wiki-Ja},
\texttt{Yahoo} and \texttt{Orkut} to evaluate their accuracy under fixed number of queries or running times.
\enlargethispage{\baselineskip}

\Cref{fig:fix_time_query} illustrates the relative error of \our and \bs under time constraints of 100ms, 200ms, 500ms, 1s, and 2s, as well as under query times of $1\times 10^6$, $2\times 10^6$, $5\times 10^6$, $1\times 10^7$, $2\times 10^7$, and $4\times 10^7$. The results are presented using the same box-plot style as \Cref{fig:overall_comparison}.
Since the \naive algorithm is based on exact butterfly counting in fixed-size subgraph, it is difficult to precisely control its running time and the total number of queries in different sizes of the input graphs. Furthermore, as analyzed in \Cref{sec:sota} and shown in \Cref{fig:overall_comparison}, \naive incurs substantially higher computational cost.
\enlargethispage{\baselineskip}
\enlargethispage{\baselineskip}

Across 4 representative datasets, both \our and \bs show decreasing relative error as time/query cost rises. However, \our consistently achieves significantly lower relative error than \bs within the same cost limitation. For instance, on the sparse graph \texttt{Amazon} with 
\enlargethispage{\baselineskip}
time, \our matches \bs's accuracy (90\% probability of relative error within 10\%, max error 15\%) in $100ms$, while \bs needs $500ms$. With constrained query count on \texttt{Wiki-Ja}, \our keeps relative error within 5\% with fewer than $1\times 10^{6}$ queries, whereas \bs, even with $4\times 10^{7}$ queries, only ensures a 20\% relative error with high probability.


\enlargethispage{\baselineskip}
Moreover, under strict computational cost constraints, such as limiting queries to $1\times 10^{6}$ or $2\times 10^{6}$ for \texttt{Amazon} and \texttt{Orkut}, \our can still produce results with reasonable accuracy, while \bs fails to produce meaningful outputs. This is because before starting the sampling rounds for butterfly counting, \bs requires at least $O(n)$ queries to obtain vertex degrees before sampling for butterfly counting, where $n$ is the number of vertices in one layer, making it ineffective when query limits are fewer than or comparable to $n$. 
\enlargethispage{\baselineskip}
In contrast, \our maintains accuracy even with fewer queries than vertices. For example, on \texttt{Orkut} with at least $2.7\times 10^{6}$ vertices per layer, \bs cannot provide any result within $2\times 10^{6}$ queries, while \our achieves a maximum error of 5\% using only $1\times 10^{6}$ queries.

}

\noindent{\bf Computational cost for {fixed} relative error.}
{In this paragraph, we compare the efficiency of different methods under aligned error conditions in \Cref{fig:cost_of_vary_density}.}
We vary the density of the graphs by retaining each edge with probability $p\in{0.2,0.4,0.6,0.8,1}$, using the same $s_1,s_2$ settings as in the main experiments and increasing $r$ until the relative error falls below $3\%$. 
{Across all four datasets, \our achieves this accuracy with shorter running time and fewer queries than \bs.} Moreover, as density grows, \our often incurs lower cost while \bs suffers higher cost, demonstrating the potential advantage of our method on dense graphs.

\noindent{\bf{Effect of varying $s_1$.}} To study the effect of $s_1$ in \our, we vary it on six graphs, showing results for \texttt{Bi-LiveJournal} and \texttt{Wiki-He} due to limited space. \Cref{fig:cost_of_vary_s1} presents performance for $s_1 \in \{c\sqrt{m} \mid c=0.1,0.2,0.5,1,2,5,10\}$. Relative errors decrease slightly as $s_1$ grows. For $s_1 \leq 0.5\sqrt{m}$, increasing $s_1$ reduces query counts while running time stays stable. For $s_1 \geq 0.5\sqrt{m}$, both queries and running time increase. Thus, we set $s_1 = 0.5\sqrt{m}$ by default.
\\
\noindent{\bf Peak memory usage. }{Now we compare the peak memory usage of \our with the baseline methods and present the results in \Cref{table:peak_memory}.
Peak memory was monitored during single-threaded execution on real datasets, and the observed trends align with our theoretical analysis. Across all 15 datasets in \Cref{table:datasets}, \our consistently exhibits the lowest peak memory. For example, on the largest dataset \texttt{Wiki-En}, \our peaks at only $0.63$ MB, while \naive and \bs require $777.39$ MB and $112.00$ MB, respectively. On the densest dataset \texttt{Bag-Ny}, \our reports $0.31$ MB, compared to $109.38$ MB for \naive and $6.00$ MB for \bs. These results demonstrate that \our achieves a clear advantage in memory efficiency.
}

\enlargethispage{\baselineskip}

\section{Related Work}
\label{sec:related}

\noindent\textbf{Butterfly counting in bipartite graphs}. 
Recent years have seen rapid progress in butterfly counting literature. Exact algorithms have been studied in memory \cite{Wang2014RectangleCI,sanei2018butterfly,Wang2018VertexPB} (with \cite{Wang2018VertexPB} being the state-of-the-art), as well as in batch-dynamic graphs \cite{Wang2022AcceleratedBC}, uncertain graphs \cite{Zhou2021ButterflyCO}, temporal graphs \cite{Cai2023EfficientTB}, parallel computing \cite{Shi2019ParallelAF, Xu2022EfficientLB}, distributed settings \cite{Weng2023DistributedAT}, and hierarchical memory \cite{Wang2023IOEfficientBC}.Note that all the exact butterfly counting algorithms must access the whole graph. Approximate algorithms have been explored in streaming graphs \enlargethispage{\baselineskip}\cite{Li2021ApproximatelyCB, SaneiMehri2018FLEETBE, Sheshbolouki2021sGrappBA} and general bipartite graphs using sampling strategies \cite{SaneiMehri2017ButterflyCI}, with \cite{zhang2023scalable} introducing one-sided weighted sampling. Our algorithm is also an estimate algorithm and has improved both theory and practicality compared to the algorithms in \cite{SaneiMehri2017ButterflyCI,zhang2023scalable}. 
\\
\noindent\textbf{Graph parameter estimation under query models}. 
The query model is the basis of sublinear algorithms and graph parameter estimation. \cite{feige2004sums, Goldreich2006ApproximatingAP} develop algorithms to estimate the average degree number with a constant level of precision.  \cite{Eden2015ApproximatelyCT} focuses on solving triangle cases and establishing their theoretical boundaries. Then, \cite{Eden2017OnAT} extends their work to address the estimation of the number of k-cliques.
{
Our work also adopts the standard query model but centers on different structures: for example, wedges in bipartite graphs, unlike triangle counting, which targets unipartite graphs without the notion of wedges.}
Other work for estimating graph parameters includes approximating the size of the minimum weight spanning tree \cite{Chazelle2001ApproximatingTM, Czumaj2004EstimatingTW} and the minimum vertex cover \cite{Hassidim2009LocalGP, Onak2011ANS} under the query model. Variants of the query model have also received extensive attention. \cite{Gonen2010CountingSA} designs the algorithm for approximating star numbers with only degree and neighbor queries. \cite{Chierichetti2016OnSN, Dasgupta2014OnET} design algorithms in a weaker query model without random node query, often referred to as ``random walk''. \cite{Eden2018FasterSA} considers the problem in low arboricity graphs.
{However, they cannot use a query model equivalent to, or stricter than, the standard query model in this paper to address the bipartite butterfly counting problem.
Thus, the existing approaches above cannot be easily used to solve our problem. To our best knowledge, we are the first work to consider estimating butterfly counts under the query model.}
\section{conclusions}
\label{sec:conclusions}
In this paper, we address the problem of approximate butterfly counting under the query model. We propose a two-level sampling algorithm, \our, designed to achieve high efficiency and accuracy. Leveraging heavy-light partitioning and guess-and-prove techniques, we theoretically demonstrate that \our achieves arbitrarily high accuracy on general bipartite graphs with sublinear time and query complexity. Extensive experiments on {15} real-world datasets confirm the effectiveness of \our in significantly reducing both query cost and running time.

\enlargethispage{\baselineskip}
\enlargethispage{\baselineskip}
\enlargethispage{\baselineskip}

\bibliographystyle{IEEEtran}
\section*{Acknowledgments}
Yuhao Zhang is supported by NSFC 62572310. Kai Wang is supported by NSFC 62302294. Kuan Yang is supported by NSFC grant 62102253 and NSFC grant 62572298.
Yuhao Zhang is the corresponding author.

\newpage
\section*{AI-GENERATED CONTENT ACKNOWLEDGEMENT}
\label{sec:ai}
In this work, we used ChatGPT-4o and DeepSeek-R1 for language polishing. We accept full responsibility for the intellectual integrity of this work and affirm its adherence to academic ethical standards.
\bibliography{reference}


\end{document}